\documentclass[UKenglish]{llncs}
\usepackage[utf8]{inputenc}

\usepackage{url}
\usepackage{cmap}
\usepackage[usenames,dvipsnames]{xcolor}
\usepackage{amsfonts,amsmath,amssymb}
\usepackage{verbatim,paralist,mathtools}

\usepackage{tikz}
\usetikzlibrary{arrows,shapes,decorations}

\usepackage{syntax}
\usepackage{mdwlist}
\usepackage{multirow}

\usepackage{listings}
\usepackage{sty/proofsInAppendix}
\usepackage{sty/programs}
\usepackage{sty/mathematics}
\usepackage{sty/descitem}
\usepackage{sty/tuple}
\usepackage{cancel}

\title{Robustness against Power is \pspace-complete}

\author{Egor Derevenetc\inst{1,2}\and Roland Meyer\inst{2}}
\institute{$^1$Fraunhofer ITWM\quad $^2$University of Kaiserslautern}

\begin{document}

\maketitle
\begin{abstract}

Power is a RISC architecture developed by IBM, Freescale, and several other companies and implemented in a series of POWER processors.
The architecture features a relaxed memory model providing very weak guarantees with respect to the ordering and atomicity of memory accesses.

Due to these weaknesses, some programs that are correct under sequential consistency (SC) show undesirable effects when run under Power.
We call these programs not robust against the Power memory model.
Formally, a program is robust if every computation under Power has the same data and control dependencies as some SC computation.

Our contribution is a decision procedure for robustness of concurrent programs against the Power memory model.
It is based on three ideas.
First, we reformulate robustness in terms of the acyclicity of a happens-before relation.
Second, we prove that among the computations with cyclic happens-before relation there is one in a certain normal form.
Finally, we reduce the existence of such a normal-form computation to a language emptiness problem.
Altogether, this yields a $\pspace$ algorithm for checking robustness against Power.
We complement it by a matching lower bound to show $\pspace$-completeness.
\end{abstract}

\section{Introduction}

To execute code as fast as possible, modern processors reorder operations.
For example, Intel~x86/x86-64 and SPARC processors implement the Total Store Ordering (TSO) memory model~\cite{Owens2009} which allows write buffering: store operations in each thread can be queued and get executed on memory later.
Processors can also execute independent instructions out of program order as soon as the input data and computational units are available for them.
This is an inherent feature of the POWER and ARM microprocessors~\cite{marangetTutorialDraft}.
Moreover, Power and ARM memory models, unlike TSO, do not guarantee store atomicity: one write can become visible to different threads at different times.
They only ensure that all threads see stores to the same memory location in the same order; stores to different memory locations can be seen in different order by different threads.

All these optimizations are usually designed so that a single-threaded program has the illusion that its instructions are executed in program order.
The picture changes in the presence of concurrency.
Concurrent programs are often assumed to have sequentially consistent (SC) semantics~\cite{Lamport79}: each thread executes its operations in program order, stores become visible immediately to all threads.
Concurrent programs may observe a difference from SC when run on a modern processor with a weak memory model.
To see this, consider the MP program in Figure~\ref{Figure:MP}.
SC and TSO forbid the situation where $\areg_1>\areg_2$ upon termination of both threads.
However, this is possible on Power: instruction $c$ can read the value written by $b$, whereas $d$ reads the initial value.

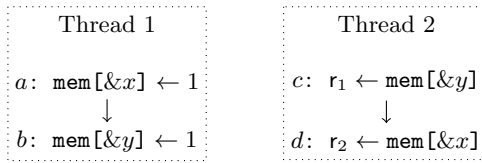
\begin{figure}[t]
\centering
\begin{tikzpicture}[nodes={rectangle,draw=none,fill=none}]
  \matrix[row sep=0.3cm,column sep=1cm,nodes={rectangle,draw=none,fill=none}] {
    \node (thread1) {Thread 1}; & \node (thread2) {Thread 2}; \\
    \node (storex) {$a\colon$ \thestore{\&x}{1}}; & \node (loady) {$c\colon$ \theload{\areg_1}{\&y}}; \\
    \node (storey) {$b\colon$ \thestore{\&y}{1}}; & \node (loadx) {$d\colon$ \theload{\areg_2}{\&x}}; \\
  };
  \draw[->] (storex) edge (storey);
  \draw[->] (loady) edge (loadx);
  \draw[dotted] (thread1.north -| storey.west) rectangle (storey.south -| storey.east);
  \draw[dotted] (thread2.north -| loady.west) rectangle (loadx.south -| loadx.east);
\end{tikzpicture}

\caption{Message Passing (MP) program~\cite{SarkarPLDI2011}.
By $\&x$ and $\&y$ we denote the addresses of the variables $x$ and $y$.
Initially, $x=y=0$.
The first thread writes a message into $x$ and sets flag variable $y$, signifying that the message is written.
The second thread reads the flag and, if it is set, expects to see the message written to $x$ by the first thread.}
\label{Figure:MP}
\end{figure}

We call a program \emph{not robust} against Power~\cite{ShashaSnir88,burckhardt-musuvathi-CAV08,Sen2011,AlglaveM11,BMM11,calin2013,bouajjani2013checking} if 
it exhibits non-SC behaviors when executed under the Power memory model. 
More formally, a program is robust if all its Power computations have the same data and control dependencies as the computations under SC.
That is, for every Power computation there is a sequentially consistent computation which executes the same instructions, all loads read from the same stores in both computations, and stores to the same address happen in the same order.
Robust programs produce the same results on Power and SC architectures, which means verification results for SC remain valid for the weak memory model.

We present an algorithm for deciding robustness against Power.
This is the first decidability result for this architecture and, more generally, the first decidability result for a non-store atomic memory model.
We obtain the algorithm in the following steps.
First, we reformulate robustness in terms of acyclicity of a happens-before relation, using the result by Shasha and Snir~\cite{ShashaSnir88}.
Second, we show that among the computations with cyclic happens-before there is always one in a certain normal form.
Next, we prove that the set of all normal-form computations can be generated by a multiheaded automaton --- an automaton model developed recently in the context of robustness~\cite{calin2013}.
Finally, to check cyclicity of the happens-before relation we intersect this automaton with regular languages.
The program is robust iff the intersection is empty.
This reduces robustness to language emptiness for multiheaded automata.
The algorithm works in space polynomial in the size of the program.
We obtain a matching lower bound by a reduction of SC-reachability to robustness, similar to~\cite{BMM11}.

\paragraph{Related work}
The happens-before relation was formulated by Lamport~\cite{lamport1978time}.
Shasha and Snir~\cite{ShashaSnir88} have shown that a computation violates sequential consistency iff it has a cyclic happens-before relation.
Burckhardt and Musuvathi~\cite{burckhardt-musuvathi-CAV08} proposed the first algorithm for detecting non-robustness against TSO based on monitoring SC computations.
Burnim et al.~\cite{Sen2011} pointed out a mistake in the definition of TSO used in \cite{burckhardt-musuvathi-CAV08} and described monitoring algorithms for the TSO and PSO memory models.
Alglave and Maranget~\cite{AlglaveM11} presented a tool to statically over-approximate happens-before cycles in programs written in x86 and Power assembly, and to insert synchronization primitives (memory fences and syncs) as required for robustness (called stability in their work).
Bouajjani et al.~\cite{BMM11} obtained the first decidability result for robustness: robustness against TSO is $\pspace$-complete for finite-state programs.
In~\cite{bouajjani2013checking} they presented a reduction of robustness against TSO to SC reachability for general programs and an algorithm for optimal fence insertion.

The Power architecture has attracted considerable recent attention.
Alglave et al.~\cite{alglave2013herding}  give an overview of the numerous publications devoted to defining its semantics.
We highlight two Power models: the operational model by Sarkar et al.~\cite{SarkarPLDI2011} and the axiomatic one by Mador-Haim et al.~\cite{mador2012axiomatic}.
These models were extensively tested against the architecture and were proven to be equivalent~\cite{mador2012axiomatic}.
Nevertheless, the operational model is known to forbid certain behaviors that are possible on real hardware\footnote{\url{http://diy.inria.fr/cats/pldi-power/\#lessvs}} and in the axiomatic model\footnote{\url{http://diy.inria.fr/cats/cav-power/}}~\cite{alglave2013herding}.
Fortunately, there is a suggested fix: in Section~4.5 of~\cite{SarkarPLDI2011} one should read \emph{from a coherence-order-earlier write} instead of \emph{from a different write} (two occurrences).
Then, the operational model is believed to strictly and tightly over-approximate Power~\cite{alglavePrivateCommunication}.
In the present paper we stick to the corrected operational model from~\cite{SarkarPLDI2011}.

Finally, we would like to note that ARM has a memory model very similar to that of Power.
The differences and similarities are highlighted by Maranget et al. in~\cite{marangetTutorialDraft,alglave2013herding}.
This fact promises a relatively easy transfer of the proof techniques used in the present paper to the ARM memory model.

\section{Programming Model}
\label{Section:Programs}
We define programs and their semantics in terms of automata.
An \emph{automaton} is a tuple
$\automaton=(\states,\alphabet,\transitions,\initialstate,\finalstates)$,
where $\states$ is a set of states,
$\alphabet$ is an alphabet,
$\transitions\subseteq\states\times(\alphabet\cup\set{\emptysequence})\times \states$ is a set of transitions,
$\initialstate\in\states$ is an initial state,
and $\finalstates\subseteq\states$ is a set of final states.
We call the automaton \emph{finite} if $\states$ and $\alphabet$ are finite.
We write $\state_1\transitionto{\letter}\state_2$ if
$\transition=(\state_1,a,\state_2)\in\transitions$ and denote
$\sourcestateOf{\transition}:=\state_1$, 
$\destinationstateOf{\transition}:=\state_2$, $\labelOf{\transition}=a$.
The \emph{language} of the automaton is 
$\langOf{\automaton}:=\setcond{\sigma\in\alphabet^*}{\initialstate\transitionto{
\sigma}\state\text{ for some }\state\in\finalstates}$.
For a sequence $\sigma=\letter_1\dots\letter_n\in\Sigma^*$ we define
$\lengthOf{\sigma}:=n$, $\sigma[i]:=\letter_i$, $\firstOf{\sigma}:=\letter_1$,
and $\lastOf{\sigma}:=\letter_n$.
We use $\cdot$ for concatenation, $\projectionOf{}{}$ for projection, and $\varepsilon$ for the empty sequence.
Given $\alpha\in\alphabet^*$ and $a,b\in\alpha$, we write $a\before{\alpha}b$ if $\alpha=\alpha_1\cdot a\cdot\alpha_2\cdot b\cdot\alpha_3$.
Given a function $f\colon X\to Y$, $x'\in X$, and $y'\in Y$, we define $f'=f[x' \hookleftarrow y']$ by $f'(x):=f(x)$ for $x\in X\setminus\set{x'}$ and $f'(x'):=y'$.

A program is a finite sequence of threads: $\program = \thread_1\ldots\thread_n$.
A \emph{thread} is an automaton $\thread_{\tid} = (\controlstates_{\tid},\commands,\instructions_{\tid},\initialcontrolstate_{\tid},\controlstates_{\tid})$ with a finite set of control states $\controlstates_{\tid}$, all of them being final, initial state $\initialcontrolstate_{\tid}$, and a set of transitions $\instructions_{\tid}$ called \emph{instructions} and labeled with \emph{commands} $\commands$ defined below.
Each thread has an id from $\tiddomain:=[1..\lengthOf{\program}]$.

Let $\datadomain=\addrdomain$ be a finite domain of values and addresses containing the value $0$.
Let $\regdomain$ be a finite set of registers that take values from $\datadomain$.
The set of commands $\commands$ includes loads, stores, local assignments, and conditionals (\lit*{assume}):
\setlength{\grammarindent}{5em}
\setlength{\grammarparsep}{\parskip}
\begin{grammar}
  <cmd> ::= <reg> $\leftarrow$ "mem["<expr>"]" \syntax{|} "mem["<expr>"]" $\leftarrow$ <expr>
  \alt <reg> $\leftarrow$ <expr> \syntax{|} "assume("<expr>")"
  %\alt "sync" \syntax{|} "lwsync" \syntax{|} "isync"
\end{grammar}
The set of expressions $\exprdomain$ is defined over constants from $\datadomain$, registers from $\regdomain$, and (unspecified) functions $\functiondomain$ over $\datadomain\cup\set{\bot}$.
We assume that these functions return $\bot$ iff any of the arguments is $\bot$.

\subsection{Power Semantics}

We briefly recall the corrected model from~\cite{SarkarPLDI2011}.
The state of a running program consists of the runtime states of threads and the state of a storage subsystem.

The runtime state of a thread includes information about the instructions being executed by the thread.
In order to start executing an instruction, the thread must \emph{fetch} it.
The thread can fetch any instruction whose source control state is equal to the destination state of the last fetched instruction.
Then, the thread must perform any computation required by the semantics of this instruction.
For example, for a load the thread must compute the address being accessed, then read the value at this address, and place it into the target register.
The last step of executing an instruction is \emph{committing} it.
Committing an instruction requires committing all its \emph{dependencies}.
For example, before committing a load the thread must commit all its \emph{address dependencies} --- the instructions which define the values of registers used in the address expression --- and \emph{control dependencies} --- the program-order-earlier (fetched earlier than the load) conditional instructions.
Moreover, all loads and stores accessing the same address must be committed in the order in which they were fetched.

The storage subsystem keeps track, for each address, of the global ordering of stores to this address --- the \emph{coherence order} --- and the last store to this address \emph{propagated} to each thread.
When a thread commits a store, this store is assigned a position in the coherence order which we identify by a rational number --- the \emph{coherence key}.
We choose rational numbers (rather than naturals) to be able to insert a store between any two stores in the coherence order.
The key must be greater than the coherence key of the last store to the same address propagated to this thread.
The committed store is immediately propagated to its own thread.
At some point later this store can be propagated to any other thread, as long as it is coherence-order-later (has a greater coherence key) than the last store to the same address propagated to that thread.
When a thread loads a value from a certain address, it gets the value written by the last store to this address propagated to this thread.
A thread can also forward the value being written by a not yet committed store to a later load reading the same address.
This situation is called an \emph{early read}.

An important property of Power is that it maintains the illusion of sequential consistency for single-threaded programs.
This means that reorderings on the thread level must not lead to situations when, e.g., a program-order-later load reads a coherence-order-earlier store than the one read by a program-order-earlier load from the same address.
In~\cite{SarkarPLDI2011} these restrictions are enforced by the mechanism of restarting operations.
We put these conditions into the requirements on final states of the running program instead.

To keep the paper readable, we omit the descriptions of Power synchronization instructions: \lit*{sync}, \lit*{lwsync}, \lit*{isync}.
All constructions in the paper can be consistently extended to support them with the final result continuing to hold.

Formally, we define the semantics of program $\program$ on Power by a \emph{Power automaton} $\powerautomaton(\program):=(\powerstates,\events,\powertransitions,\initialpowerstate,\finalpowerstates)$. 
Here, $\events$ is a set of labels called \emph{events} that we define together with the transitions.

\subsubsection{State space}
A state of the Power automaton is a pair $\powerstate=(\rtstatesmap,\ststate)\in\powerstates$ with runtime thread states $\rtstatesmap\colon\tiddomain\to\rtstates$ and storage subsystem state $\ststate\in\ststates$.

A runtime thread state $\rtstate=(\fetched,\committed,\loaded)\in\rtstates$ includes
a finite sequence of fetched instructions $\fetched\in\instructions^*$,
a set of indices of committed instructions $\committed\subseteq[1..\lengthOf{\fetched}]$,
and a function giving the store read by a load $\loaded\colon[1..\lengthOf{\fetched}]\to\set{\bot}\cup\setcond{\initialstore_{\anaddr}}{\anaddr\in\addrdomain}\cup\tiddomain\times\naturalnumbers$. 
We use $\initialstore_{\anaddr}$ to denote the initial store of value $0$ to address $\anaddr$.
The initial state of a running thread is $\initialrtstate:=(\emptysequence,\emptyset,\lambda\anindex.\bot)$.

A state of the storage subsystem $\ststate=(\coherence,\propagated)\in\ststates$ includes
a mapping from a store instruction (its thread id and index in the list of fetched instructions) to its position in the coherence order $\coherence\colon(\tiddomain\times\naturalnumbers\cup\setcond{\initialstore_{\anaddr}}{\anaddr\in\addrdomain})\to\rationalnumbers$,
and a mapping from a thread id and an address to the last store to this address propagated to this thread $\propagated\colon\tiddomain\times\addrdomain\to\setcond{\initialstore_{\anaddr}}{\anaddr\in\addrdomain}\cup\tiddomain\times\naturalnumbers$.
The initial state of the storage subsystem is $\initialststate:=(\lambda\tid.\lambda\anindex.0,\lambda\tid.\lambda\anaddr.\initialstore_{\anaddr})$.

The initial state of automaton $\powerautomaton(\program)$ is $\initialpowerstate:=(\lambda\tid.\initialrtstate,\initialststate)$.

\subsubsection{Transition relation}
Fix a state $\powerstate=(\rtstatesmap,\ststate)$ with $\ststate=(\coherence,\propagated)$ and a thread id $\tid\in\tiddomain$ with runtime state $\rtstatesmap(\tid)=(\fetched,\committed,\loaded)$.

Let $\evaluate(\tid,\anindex,\anexpr)$ return the value in $\datadomain$ of expression $\anexpr$ in the $\anindex$'th fetched instruction of thread $\tid$, or $\bot$ when the value is undefined.
Formally $\evaluate(\tid,\anindex,\anexpr):=\aval$, where $\aval$ is computed as follows.
If $\anexpr\in\datadomain$, then $\aval:=\anexpr$.
If $\anexpr=\afun(\anexpr_1\ldots\anexpr_n)$, then $\aval:=\afun(\evaluate(\tid,\anindex,\anexpr_1)\ldots\evaluate(\tid,\anindex,\anexpr_n))$.
Otherwise, $\anexpr=\areg\in\regdomain$.
Let $\anindex'\in[1..\anindex-1]$ be the greatest index, such that $\fetched[\anindex']$ is a local assignment or a load to $\areg$.
If there is no such index, we define $\aval:=0$.
If $\labelOf{\fetched[\anindex']}=\theassign{\areg}{\anexpr_{\aval}}$, then $\aval:=\evaluate(\tid,\anindex',\anexpr_{\aval})$.
If $\labelOf{\fetched[\anindex']}=\theload{\areg}{\anexpr_{\anaddr}}$, then $\aval:=\bot$ if $\loaded[\anindex']=\bot$, $\aval:=0$ if $\loaded[\anindex']=\initialstore_{*}$, and $\aval:=\getvalue(\loaded[\anindex'])$ otherwise (see the definition of $\getvalue$ below).

The expression $\getaddr(\tid,\anindex)$ returns the value of the address argument of the $\anindex$'th fetched instruction of thread $\tid$ and is defined as follows.
We use the special value $\top$ if the instruction has no such argument.
If $\labelOf{\fetched[\anindex]}=\theload{\areg}{\anexpr_{\anaddr}}$ or $\labelOf{\fetched[\anindex]}=\theload{\anexpr_{\anaddr}}{\anexpr_{\aval}}$, then $\getaddr(\tid,\anindex):=\evaluate(\tid,\anindex,\anexpr_{\anaddr})$.
Otherwise, $\getaddr(\tid,\anindex):=\top$.

Similarly, the expression $\getvalue(\tid,\anindex)$ returns the value of the value argument of the $\anindex$'th fetched instruction of thread $\tid$ and is defined as follows.
If $\labelOf{\fetched[\anindex]}=\thestore{\anexpr_{\anaddr}}{\anexpr_{\aval}}$, $\labelOf{\fetched[\anindex]}=\theassign{\areg}{\anexpr_{\aval}}$, or $\labelOf{\fetched[\anindex]}=\theassume{\anexpr_{\aval}}$, then $\getvalue(\tid,\anindex)=\evaluate(\tid,\anindex,\anexpr_{\aval})$.
Otherwise, $\getvalue(\tid,\anindex):=\top$.

The expressions $\addrdeps(\tid,\anindex)$, $\datadeps(\tid,\anindex)$, $\ctrldeps(\tid,\anindex)$ denote the sets of indices of instructions in thread $\tid$ being respectively address, data, and control dependencies of the $\anindex$'th instruction.
The first two can be formally defined in a recursive manner, similar to $\evaluate$.
Also, $\ctrldeps(\tid,\anindex):=\setcond{\anindex'\in[1..\anindex-1]}{\labelOf{\fetched[\anindex']}=\theassume{\anexpr_{\aval}}}$.

Let $\thread_{\tid}=(\controlstates_{\tid},\commands,\instructions_{\tid},\initialcontrolstate_{\tid},\controlstates_{\tid})\in\program$.
The transition relation $\powertransitions$ is the smallest relation defined by the rules below:
\begin{description}
\descitem{POW-FETCH}
Consider $\instruction\in\instructions_{\tid}$ with $\sourcestateOf{\instruction}=\destinationstateOf{\lastOf{\fetched}}$ or $\sourcestateOf{\instruction}=\initialcontrolstate_{\tid}$ if $\fetched=\emptysequence$, then:
$$(\rtstatesmap,\ststate)\transitionto{(\fetchkind,\tid,\instruction)}(\rtstatesmap[\tid\hookleftarrow(\fetched\cdot\instruction,\committed,\loaded)],\ststate).$$

\descitem{POW-LOAD}
If $\fetched[\anindex]$ is a load, $\loaded[\anindex]=\bot$, $\anaddr=\getaddr(\tid,\anindex)\neq\bot$, then:
{\small $$(\rtstatesmap,\ststate)\transitionto{(\loadkind,\tid,\anindex,\anaddr)}(\rtstatesmap[\tid\hookleftarrow(\fetched,\committed,\loaded[\anindex\hookleftarrow\propagated(\tid,\anaddr)])],\ststate).$$}

\descitem{POW-EARLY}
Let $\fetched[\anindex]$ be a load, $\loaded[\anindex]=\bot$, and $\anaddr=\getaddr(\tid,\anindex)\neq\bot$.
Let $\anindex'\in[1..\anindex-1]$ be the greatest index such that $\fetched[\anindex']$ is a store with $\anaddr'=\getaddr(\tid,\anindex')\in\set{\anaddr,\bot}$.
If $\anaddr'\neq\bot$, $\getvalue(\tid,\anindex')\neq\bot$, $\anindex'\not\in\committed$, then:
$$(\rtstatesmap,\ststate)\transitionto{(\loadkind,\tid,\anindex,\anaddr)}(\rtstatesmap[\tid\hookleftarrow(\fetched,\committed,\loaded[\anindex\hookleftarrow(\tid,\anindex')])],\ststate).$$

\descitem{POW-COMMIT}
Consider $\anindex\in[1..\lengthOf{\fetched}]\setminus\committed$ where $\fetched[\anindex]$ is not a store.
Assume $\addrdeps(\tid,\anindex)\cup\datadeps(\tid,\anindex)\cup\ctrldeps(\tid,\anindex)\subseteq\committed$.
Assume $\anaddr=\getaddr(\tid,\anindex)\neq\bot$ and $\aval=\getvalue(\tid,\anindex)\neq\bot$.
If $\anaddr\neq\top$, assume $\setcond{\anindex'\in[1..\anindex-1]}{\getaddr(\tid,\anindex')\in\set{\anaddr,\bot}}\subseteq\committed$.
In case $\fetched[\anindex]$ is a load, assume $\loaded[\anindex]\neq\bot$.
In case $\fetched[\anindex]$ is an $\theassume{}$, assume $\aval\neq 0$.
Then:
$$(\rtstatesmap,\ststate)\transitionto{(\commitkind,\tid,\anindex)}(\rtstatesmap[\tid\hookleftarrow(\fetched,\committed\cup\set{\anindex},\loaded)],\ststate).$$

\descitem{POW-STORE}
Assume all the preconditions from the previous rule hold, but $\fetched[\anindex]$ is a store.
Choose a coherence key $\coherencekey\in\rationalnumbers$ such that there is no $\tid'\in\tiddomain$, $\anindex'\in\naturalnumbers$ for which $\coherence(\tid',\anindex')=\coherencekey$.
Then:
$$(\rtstatesmap,\ststate)\transitionto{(\commitkind,\tid,\anindex,\coherencekey,\anaddr)}(\rtstatesmap[\tid\hookleftarrow(\fetched,\committed\cup\set{\anindex},\loaded)],\ststate'),$$
where $\ststate':=(\coherence[(\tid,\anindex)\hookleftarrow\coherencekey],\propagated)$.

Additionally, this transition is immediately followed by a \descref{POW-PROP} transition propagating the store to the thread where it was committed.

\descitem{POW-PROP}
Consider $\tid'\in\tiddomain$, $\anindex'\in\naturalnumbers$ with $\coherence(\tid',\anindex')\neq\bot$.
Let $\anaddr=\getaddr(\tid',\anindex')$.
Assume $\coherence(\propagated(\tid,\anaddr))<\coherence(\tid',\anindex')$.
Then:

$$(\rtstatesmap,\ststate)\transitionto{(\propagatekind,\tid,\tid',\anindex',\anaddr)}(\rtstatesmap,(\coherence,\propagated[(\tid,\anaddr)\hookleftarrow(\tid',\anindex')])).$$
\end{description}

\subsubsection{Final states}
The set of final states $\finalpowerstates\subseteq\powerstates$ consists of all states $\powerstate=(\rtstatesmap,(\coherence,\propagated))\in\powerstates$, such that for each $\tid\in\tiddomain$, $\rtstatesmap[\tid]=(\fetched,\committed,\loaded)$ the following holds:
\begin{description}
\descitem{FIN-COMM} All instructions are committed: $\committed=[1..\lengthOf{\fetched}]$.
\descitem{FIN-LD} Loads agree with the coherence order. Let $\fetched[\anindex]$ be a load, and $\fetched[\anindex']$ be an earlier load to the same address: $\anindex'<\anindex$, $\getaddr(\tid,\anindex)=\getaddr(\tid,\anindex')$.
Then $\coherence(\loaded[i'])\leq\coherence(\loaded[i])$.
\descitem{FIN-LD-ST} Loads and stores in the same thread agree with the coherence order. Let $\fetched[\anindex]$ be a load, let $\fetched[\anindex']$ be an earlier store to the same address: $\anindex'<\anindex$, $\getaddr(\tid,\anindex)=\getaddr(\tid,\anindex')$.
Then $\coherence(\tid,\anindex')\leq\coherence(\loaded[\anindex])$.
\end{description}
The set of all \emph{Power computations of program $\program$} is $\computationsOf{\program}{\power}:=\langOf{\powerautomaton(\program)}$.
The set of all \emph{SC computations of the program} $\computationsOf{\program}{\seqcons}\subseteq\computationsOf{\program}{\power}$ includes only those computations where each instruction is executed atomically, and stores are immediately propagated to all threads.

\begin{example}\label{Example:Computation}
$\sigma_{\mathit{MP}}=\fetchkind(a)\cdot\commitkind(a)\cdot\propagatekind(a,1)\cdot\fetchkind(b)\cdot\commitkind(b)\cdot\propagatekind(b,1)\cdot\propagatekind(b,2)\cdot\fetchkind(c)\cdot\fetchkind(d)\cdot\loadkind(c)\cdot\loadkind(d)\cdot\commitkind(d)\cdot\commitkind(c)$ is a feasible Power computation of the program MP (Figure~\ref{Figure:MP}):
\begin{itemize}
\item $\fetchkind(a):=(\fetchkind,1,a)$ --- thread 1 fetches store instruction $a$.
\item $\commitkind(a):=(\commitkind,1,1,1,\&x)$ --- thread 1 commits $a$ with $\coherencekey=1$.
\item $\propagatekind(a,1):=(\propagatekind,1,1,1,\&x)$ --- $a$ is propagated to its own thread.
\item $\fetchkind(b):=(\fetchkind,1,b)$ --- thread 1 fetches store instruction $b$.
\item $\commitkind(b):=(\commitkind,1,2,2,\&y)$ --- thread 1 commits $b$ with $\coherencekey=2$.
\item $\propagatekind(b,1):=(\propagatekind,1,1,2,\&x)$ --- the store is propagated to its thread.
\item $\propagatekind(b,2):=(\propagatekind,2,1,2,\&x)$ --- the store is propagated to thread 2.
\item $\fetchkind(c):=(\fetchkind,2,c)$ --- thread 2 fetches load $c$.
\item $\fetchkind(d):=(\fetchkind,2,c)$ --- thread 2 fetches load $d$.
\item $\loadkind(c):=(\loadkind,2,1,\&y)$ --- thread 2 reads value $1$ written by $b$ to $y$, because $b$ was propagated to thread 2.
\item $\loadkind(d):=(\loadkind,2,2,\&x)$ --- thread 2 reads the initial value $0$ of $x$, because $a$ was not propagated to thread 2.
\item $\commitkind(d):=(\commitkind,2,2)$ --- thread 2 commits load $d$.
\item $\commitkind(c):=(\commitkind,2,1)$ --- thread 2 commits load $c$.
\end{itemize}

In the end, \descref{FIN-COMM} holds as all fetched instructions are indeed committed, and \descref{FIN-LD} and \descref{FIN-LD-ST} trivially hold, as none of the threads has two instructions accessing the same address.
\end{example}

\begin{lemma}\label{Lemma:FinalStateIsDeterminedByComputation}
Assume $\initialpowerstate\transitionto{\sigma}\powerstate\in\finalpowerstates$.
Then $\powerstate$ is uniquely determined.
\end{lemma}
\begin{proof}
Given a state and an event $\event$, there is at most one transition from this state labeled by $\event$.
This is clear for non-$\loadkind$ events.
For $\loadkind$ events, this follows from Lemma~\ref{Lemma:EarlyReadLooksLikeThis} and Lemma~\ref{Lemma:LoadFromMemoryLooksLikeThis}: if a $\loadkind$ event was produced by a load from memory transition, then condition (3) from Lemma~\ref{Lemma:LoadFromMemoryLooksLikeThis} holds, then condition (1) from Lemma~\ref{Lemma:EarlyReadLooksLikeThis} cannot hold for any store, therefore, the $\loadkind$ event cannot be produced by an early read transition.
\qed
\end{proof}

\begin{lemma}\label{Lemma:SourceDoesNotChange}
Let $\initialpowerstate\transitionto{\sigma}(\rtstatesmap,\ststate)\transitionto{\event}(\rtstatesmap',\ststate')$.
Let $(\fetched,\committed,\loaded)=\rtstatesmap(\tid)$, $(\fetched',\committed',\loaded')=\rtstatesmap'(\tid)$ for some $\tid\in\tiddomain$.
If $\loaded[\anindex]\neq\bot$, then $\loaded'[\anindex]=\loaded[\anindex]$.
\end{lemma}
\begin{proof}
Follows from the $\loaded[\anindex]=\bot$ requirement in \descref{POW-LOAD} and \descref{POW-EARLY} transitions.
\qed
\end{proof}

\begin{lemma}\label{Lemma:OnceComputedDoesNotChange}
Let $\initialpowerstate\transitionto{\sigma}\powerstate\transitionto{\event}\powerstate'$.
Assume $\evaluate(\tid,\anindex,\anexpr)=\aval\neq\bot$ in $\powerstate$.
Then $\evaluate(\tid,\anindex,\anexpr)=\aval$ in $\powerstate'$.
\end{lemma}
\begin{proof}
By definition of $\evaluate$, Lemma~\ref{Lemma:SourceDoesNotChange}, and the fact that functions in $\functiondomain$ are deterministic.
\qed
\end{proof}

\begin{lemma}\label{Lemma:EarlyReadLooksLikeThis}
Consider a computation $\sigma\in\computationsOf{\program}{\power}$.
Then a load $(\tid,\anindex)$ reads a value from a store $(\tid,\anindex')$ via an early read (\descref{POW-EARLY}) transition iff
(1) $\sigma=\sigma_1\cdot(\loadkind,\tid,\anindex,\anaddr)\cdot\sigma_2\cdot(\commitkind,\tid,\anindex',*,\anaddr)\cdot\sigma_3$, $\anindex'\in\intrange{1}{\anindex-1}$
and (2) $\sigma_3$ does not contain events matching $(\commitkind,\tid,\intrange{\anindex'+1}{\anindex-1},*,\anaddr)$.
\end{lemma}
\begin{proof}
From left to right.
Assume the load $(\tid,\anindex)$ reads the store $(\tid,\anindex')$ via an early read transition.
Then $(\tid,\anindex)$ must be the latest store to the same address in thread $\tid$ and must not be committed before $\loadkind$ (i.e. committed after it), therefore (1) holds.
If (2) does not hold, then $(\tid,\anindex')$ is not the latest store to address $\anaddr$ in thread $\tid$ before the $\loadkind$ event, since stores to the same address are committed in the order of fetching.
Contradiction.

From right to left.
Let $\initialpowerstate\transitionto{\sigma_1}\powerstate=(\rtstatesmap,\ststate)$.
Consider $\rtstatesmap(\tid)=(\fetched,\committed,\loaded)$.
Let $\anindex''<\anindex$ be the greatest index, such that $\fetched[\anindex'']$ is a store, $\getaddr(\anindex'')\in\set{\anaddr,\bot}$.

Assume $\anindex'<\anindex''$.
If $\getaddr(\anindex'')=\anaddr$, we get a contradiction to (2), since stores to the same address are committed in the order of fetching.
If $\getaddr(\anindex'')=\bot$, then an early read is not possible in state $\powerstate$, and the load reads from the latest propagated store (\descref{POW-LOAD}), which is coherence-order-before the store $(\tid,\anindex')$, which is program-order-before $(\tid,\anindex)$.
This situation is forbidden by \descref{FIN-LD-ST}.

By Lemma~\ref{Lemma:OnceComputedDoesNotChange}, $\getaddr(\tid,\anindex')\in\set{\anaddr,\bot}$, therefore, $\anindex''=\anindex'$.
Assume $\getaddr(\tid,\anindex')=\bot$ or $\getvalue(\tid,\anindex')=\bot$.
Then, again, a load from the latest propagated store takes place, which is impossible (see above).
Therefore, $\getaddr(\tid,\anindex')=\anaddr$ and $\getvalue(\tid,\anindex')\neq\bot$.

Obviously, $\anindex'\not\in\committed$ holds, as each fetched instruction is committed only once, and $(\tid,\anindex')$ is committed after the load takes place, see (1).
All in all, all requirements for the early read from $(\tid,\anindex')$ are met, therefore, an early read transition from state $\powerstate$ is possible.
As shown above, a load from memory transition from the same state leads to $\sigma\not\in\computationsOf{\program}{\power}$, therefore, $(\tid,\anindex)$ reads from the store $(\tid,\anindex')$ via an early read transition.
\qed
\end{proof}

\begin{lemma}\label{Lemma:LoadFromMemoryLooksLikeThis}
Consider a computation $\sigma\in\computationsOf{\program}{\power}$.
Then a load $(\tid,\anindex)$ reads a value from a store $(\tid',\anindex')$ via a load from memory (\descref{POW-LOAD}) transition iff
(1) $\sigma=\sigma_1\cdot(\propagatekind,\tid,\tid',\anindex',\anaddr)\cdot\sigma_2\cdot(\loadkind,\tid,\anindex,\anaddr)\cdot\sigma_3$,
(2) $\sigma_2$ does not contain events matching $(\propagatekind,\tid,*,*,\anaddr)$,
and (3) $\sigma_3$ does not contains events matching $(\commitkind,\tid,[1..\anindex-1],*,\anaddr)$.
\end{lemma}
\begin{proof}
From left to right.
Assume the load $(\tid,\anindex)$ reads the store $(\tid',\anindex')$ via a load from memory transition.
Then, the load has read from the latest store to address $\anaddr$ propagated to thread $\tid$, i.e., (1) and (2) hold.
Assume (3) does not hold --- $\sigma_3$ contains a commit $(\commitkind,\tid,\anindex'',*,\anaddr)$ and $\anindex''<\anindex$.
Then, $(\tid,\anindex)$ reads from the store $(\tid',\anindex')$, which is coherence-order-before the store $(\tid,\anindex'')$, which is program-order-before $(\tid,\anindex)$.
This situation is forbidden by \descref{FIN-LD-ST}.

From right to left.
By (1), (3), and Lemma~\ref{Lemma:EarlyReadLooksLikeThis}, the $\loadkind$ event was not generated by an early read transition.
Therefore, the event was generated by a load from memory transition, and the load has taken the value from the latest propagated store to address $\anaddr$, which is, by (1) and (2), $(\tid',\anindex')$.
\qed
\end{proof}

\section{Robustness}
\label{Section:Robustness}
Intuitively, a \emph{trace} $\traceOf{\sigma}$ abstracts a program computation $\sigma$ to the dataflow and control-flow relations between instructions.
Formally, the trace of $\sigma$ is a directed graph $\traceOf{\sigma}:=(\nodes,\programorder,\coherenceorder,\sourceorder,\conflictorder)$ with nodes $\nodes$ and four kinds of arcs.
The nodes are instructions together with their thread identifiers and serial numbers (in order to distinguish instructions executed in different threads and the same instruction executed multiple times in the same thread): $\nodes\subseteq(\setcond{\initialstore_{\anaddr}}{\anaddr\in\addrdomain}\cup\bigcup_{\tid\in\tiddomain}\set{\tid})\times\naturalnumbers\times\instructions_{\tid}$.
The \emph{program order} $\programorder$ is the order in which instructions were fetched in each thread.
The \emph{coherence order} $\coherenceorder$ gives the global ordering of writes to each address.
The \emph{source order} $\sourceorder$ shows the store from which a load took its value.
The \emph{conflict order} $\conflictorder$ shows, for a load, the stores to the same address following the store the load took its value from.
We define the \emph{happens-before} relation as $\happensbefore:=\programorder\cup\coherenceorder\cup\sourceorder\cup\conflictorder$.

Formally, consider a computation $\sigma\in\computationsOf{\program}{\power}$.
Let $\initialpowerstate\transitionto{\sigma}\powerstate$ with $\powerstate=(\rtstatesmap,(\coherence,\propagated))$.
By Lemma~\ref{Lemma:FinalStateIsDeterminedByComputation}, $\powerstate$ is uniquely determined.
The trace $\traceOf{\sigma}:=(\nodes,\programorder,\coherenceorder,\sourceorder,\conflictorder)$ is defined as follows.
Assuming $\tid\in\tiddomain$, $\rtstatesmap(\tid)=(\fetched,\committed,\loaded)$, $\anindex\in[1..\lengthOf{\fetched}]$, and similarly for $\tid'$, we have:
\begin{align*}
\nodes:=&\setcond{(\tid,\anindex,\fetched[\anindex])}{\tid\in\tiddomain,\ \anindex\in\naturalnumbers},\\
\programorder:=&\setcond{((\tid,\anindex,\fetched[\anindex]),(\tid,\anindex+1,\fetched[\anindex+1]))}{\\&\hspace{1.5em}\anindex\in[1..\lengthOf{\fetched}-1]},\\
\coherenceorder:=&\setcond{((\tid,\anindex,\fetched[\anindex]),(\tid',\anindex',\fetched[\anindex']))}{\\&\hspace{1.5em}\getaddr(\tid,\anindex)=\getaddr(\tid',\anindex')\text{ and }\coherence(\tid,\anindex)<\coherence(\tid',\anindex')}\ \cup\\&\setcond{(\initialstore_a,(\tid',\anindex',\fetched[\anindex']))}{a=\getaddr(\tid',\anindex')},\\
\sourceorder:=&\setcond{((\tid,\anindex,\fetched[\anindex]),(\tid',\anindex',\fetched'[\anindex']))}{\\&\hspace{1.5em}(\tid,\anindex)=\loaded'[\anindex']}\ \cup\\&\setcond{(\initialstore_a,(\tid',\anindex',\fetched'[\anindex']))}{\initialstore_a=\loaded'(\anindex')},\\
\conflictorder:=&\setcond{(a,b)}{\exists c\colon c\sourceorder a\text{ and }c\coherenceorder{}b}.
\end{align*}

We will also need address $\addrdep$ and data $\datadep$ dependence relations (defined as expected based on $\addrdeps$ and $\datadeps$).
\begin{align*}
\addrdep:=&\setcond{((\tid,\anindex,\fetched[\anindex]),(\tid,\anindex',\fetched'[\anindex]))}{\anindex\in\addrdeps(\tid,\anindex')},\\
\datadep:=&\setcond{((\tid,\anindex,\fetched[\anindex]),(\tid,\anindex',\fetched'[\anindex]))}{\anindex\in\datadeps(\tid,\anindex')}.
\end{align*}
Since $\programorder$ includes all the information from the $\fetched$ component of a thread state, $\addrdep$ and $\datadep$ can be reconstructed from $\programorder$ 
by inspecting the instructions labeling a node. 
They are therefore not included in the trace explicitly.

The \emph{robustness problem} is, given a program $\program$, to check whether the set of all traces under Power is a subset of all traces under SC: $\tracesOf{\program}{\power}\subseteq\tracesOf{\program}{\seqcons}$, where $\tracesOf{\program}{\memorymodel}:=\setcond{\traceOf{\sigma}}{\sigma\in\computationsOf{\program}{\memorymodel}}$ for $\memorymodel\in\set{\power,\seqcons}$.

Shasha and Snir have shown that a trace belongs to an SC computation iff its happens-before relation is acyclic:

\begin{lemma}[\cite{ShashaSnir88}]\label{Lemma:ShashaSnir}
A program $\program$ is robust against Power iff there is no trace $\trace\in\tracesOf{\program}{\power}$ with cyclic $\happensbefore$.
\end{lemma}

\begin{figure}[t]
\centering
\begin{tikzpicture}[nodes={rectangle,draw=none,fill=none}]
  \matrix[row sep=0.5cm,column sep=1cm,nodes={rectangle,draw=none,fill=none}] {
    & \node (thread1) {Thread 1}; & \node (thread2) {Thread 2}; \\
    \node (initialstorex) {$\initialstore_{\&x}$}; & \node (storex) {$a\colon$ \thestore{\&x}{1}}; & \node (loadx) {$d\colon$ \theload{\areg_2}{\&x}}; \\
    \node(initialstorey) {$\initialstore_{\&y}$}; & \node (storey) {$b\colon$ \thestore{\&y}{1}}; & \node (loady) {$c\colon$ \theload{\areg_1}{\&y}}; \\
  };
  \draw[->] (storex) edge node[midway,left] {$po$} (storey);
  \draw[->] (loady) edge node[midway,left] {$po$} (loadx);
  \draw[->] (storey) edge node[near end,above] {$src$} (loady);
  \draw[->] (loadx) edge node[near end,below] {$cf$} (storex);
  \draw[->] (initialstorex) edge[bend left=15] node[very near end,above] {$src$} (loadx);
  \draw[->] (initialstorex) edge node[near end,below] {$co$} (storex);
  \draw[->] (initialstorey) edge node[near end,above] {$co$} (storey);

  \draw[dotted] (thread1.north -| storey.west) rectangle (storey.south -| storey.east);
  \draw[dotted] (thread2.north -| loady.west) rectangle (loady.south -| loady.east);
\end{tikzpicture}

\caption{Trace of computation $\sigma_{\mathit{MP}}$ from Example~\ref{Example:Computation}.}
\label{Figure:TraceMP}
\end{figure}
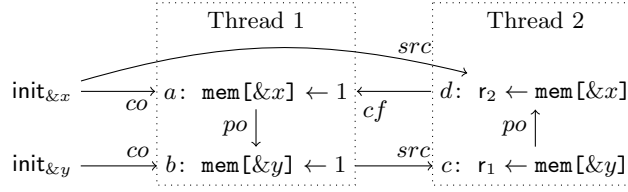

\begin{example}\label{Example:TraceMP}
The trace of computation $\sigma_{\mathit{MP}}$ (Figure~\ref{Figure:TraceMP}) has a cyclic happens-before relation.
By Lemma~\ref{Lemma:ShashaSnir}, this means that the program is not robust.
Indeed, in no SC computation load $d$ can read $0$ whereas $c$ has read $1$.
\end{example}

\section{Normal-Form Computations}
\label{Sections:NormalFormComputations}

We say that a computation $\tau\in\computationsOf{\program}{\power}$ is \emph{in normal form of degree $n$} if there is a partitioning $\tau=\tau_1\cdots\tau_n$, such that all $\fetchkind$ events are in $\tau_1$ (\descref{NF-A}) and events related to different instructions occur in different parts of the computation in the same order (\descref{NF-B}):
\begin{description}
\descitem{NF-A} $\projectionOf{(\tau_2\cdots\tau_n)}{\fetchkind}=\emptysequence$.
\descitem{NF-B} Formally, for $j\in\set{1,2}$ let $\event_j$, $\event_j'$ be events related to instruction $(\tid_j,\anindex_j)$.
If $\event_1,\event_2\in\tau_s$ and $\event_1',\event_2'\in\tau_{s'}$, then $\event_1\before{\tau_s}\event_2$ iff $\event_1'\before{\tau_{s'}}\event_2'$.
\end{description}
In the rest of this section we prove the following theorem:
\begin{theorem}\label{Theorem:NormalFormComputationsAreEnough}
A program is robust iff it has no normal-form computations of degree $\cardinalityOf{\program}+3$ with cyclic happens-before relation.
\end{theorem}

Consider a computation $\sigma\in\computationsOf{\program}{\memorymodel}$.
By $\sigma\setminus(\tid,\anindex)$ we denote the computation obtained from $\sigma$ by deleting all events related to the $\anindex$'th fetched instruction in thread $\tid$.

\begin{lemma}\label{Lemma:Cancellation}
Consider a non-empty computation $\sigma\in\computationsOf{\program}{\power}$.
Then there is a $(\tid_{\deleted},\anindex_{\deleted})$, such that $\sigma'=\sigma\setminus(\tid_{\deleted},\anindex_{\deleted})$ satisfies $\lengthOf{\sigma'}<\lengthOf{\sigma}$ and $\sigma'\in\computationsOf{\program}{\power}$.
\end{lemma}
\begin{proof}
Consider the last fetched instruction in each thread.
If among such instructions there is a non-store instruction, delete it: its result cannot be used by any other instruction.
If all these instructions are stores, delete the one, on which
(1) no load or store depends via $(\sourceorder\cup\datadep)^+\cdot\addrdep$, and
(2) no condition depends via $(\sourceorder\cup\datadep)^+$.

Towards a contradiction, assume there is no such store.
Consider a last fetched (store) instruction in a thread $\tid_1$: $(\tid_1,\anindex_1)$.
Case~1: there is a load or a store $(\tid_2,\anindex_2')$ whose address depends on $(\tid_1,\anindex_1)$.
Case~2: there is a condition $(\tid_2,\anindex_2')$ whose value depends on $(\tid_1,\anindex_1)$.
Consider the last fetched instruction in thread $\tid_2$: $(\tid_2,\anindex_2)$.
It must be a store, and it must have been committed after $(\tid_1,\anindex_1)$: a store can only be committed after all loads and stores fetched before it have their addresses determined (Case~1) and after all preceding conditions are committed (Case~2).

Continuing the reasoning, for any last fetched instruction in a thread $(\tid_j,\anindex_j)$ there is a last instruction in a different thread $(\tid_{j+1},\anindex_{j+1})$ which must have been committed later.
Taking into account finiteness of the number of threads, we get a contradiction.
\qed
\end{proof}

Fix a program $\program$.
Consider a shortest Power computation $\alpha\in\computationsOf{\program}{\power}$ with cyclic $\happensbefore$.
Let $(\tid_{\deleted},\anindex_{\deleted})$ be the instruction determined by Lemma~\ref{Lemma:Cancellation}.
Let $\alpha:=\alpha_1\cdot \deleted_1\cdot\alpha_2\cdot \deleted_2\cdots\alpha_n$, where $\set{\deleted_1\ldots \deleted_{n-1}}$ are the events related to the $\anindex_{\deleted}$'th instruction fetched in thread $\tid_{\deleted}$.
Then $\alpha\setminus(\tid_{\deleted},\anindex_{\deleted}):=\alpha':=\alpha_1\cdot\alpha_2\cdots\alpha_n$.
Since $\alpha'$ is shorter than $\alpha$, its $\happensbefore$ is acyclic.
Therefore, there is a computation $\beta\in\computationsOf{\program}{\seqcons}$ with $\traceOf{\beta}=\traceOf{\alpha'}$.

Computations $\beta$ and $\alpha'$ consist of the same fetch, load, and commit events: fetch events are determined by $\programorder$; address component $\anaddr$ of load and store commit events is determined by $\addrdep$, $\datadep$ (derivable from $\programorder$), and $\sourceorder$; since $\coherenceorder$ is the same for both computations, we can assume that matching store commit events have the same value of coherence key $\coherencekey$.
Notably, $\beta$ can have more propagate events than $\alpha'$ as Power semantics does not guarantee that all stores are propagated to all threads.
Now we reorder events in each part $\alpha_j$ of $\alpha$ in the way they follow in $\beta$.
This gives a computation $\gamma:=\projectionOf{\beta}{\alpha_1}\cdot \deleted_1\cdot\projectionOf{\beta}{\alpha_2}\cdot \deleted_2\cdots\projectionOf{\beta}{\alpha_n}$.
In the rest of the section we show that $\gamma$ is a valid Power computation of program $\program$ and has the same trace as $\alpha$.

\begin{lemma}\label{Lemma:ProgramOrderIsTheSame}
For all $\tid\in\tiddomain$ holds $\projectionOf{\projectionOf{\alpha}{\fetchkind}}{\tid}=\projectionOf{\projectionOf{\gamma}{\fetchkind}}{\tid}$.
\end{lemma}
\begin{proof}
Since $\traceOf{\beta}=\traceOf{\alpha'}$, by definition of $\alpha$ and properties of projection, for any $\tid\in\tiddomain$ we have
\begin{align*}
\projectionOf{\projectionOf{\alpha}{\fetchkind}}{\tid}&=\projectionOf{{\projectionOf{\alpha_1}{\fetchkind}}}{\tid}\cdot\projectionOf{\projectionOf{\deleted_1}{\fetchkind}}{\tid}\cdots\projectionOf{{\projectionOf{\alpha_n}{\fetchkind}}}{\tid}\\
&=\projectionOf{\cdots(\projectionOf{\projectionOf{\beta}{\fetchkind}}{\tid})}{(\projectionOf{\projectionOf{\alpha_i}{\fetchkind}}{\tid})}\cdot\projectionOf{\projectionOf{\deleted_i}{\fetchkind}}{\tid}\cdots\\
&=\projectionOf{\projectionOf{\projectionOf{\beta}{\alpha_1}}{\fetchkind}}{\tid}\cdot\projectionOf{\projectionOf{\deleted_1}{\fetchkind}}{\tid}\cdots\projectionOf{\projectionOf{\projectionOf{\beta}{\alpha_n}}{\fetchkind}}{\tid}\\
&=\projectionOf{\projectionOf{(\projectionOf{\beta}{\alpha_1}\cdot \deleted_1\cdots{\projectionOf{\beta}{\alpha_n}})}{\fetchkind}}{\tid}\\
&=\projectionOf{\projectionOf{\gamma}{\fetchkind}}{\tid}.
\end{align*}
\qed
\end{proof}

\begin{lemma}\label{Lemma:SourceRelationIsTheSame}
Consider some $(\tid,\anindex)$ and $(\tid',\anindex')$.
Let $\predicate(\sigma):=\true$ if requirements (1)--(2) from Lemma~\ref{Lemma:EarlyReadLooksLikeThis} or (1)--(3) from Lemma~\ref{Lemma:LoadFromMemoryLooksLikeThis} hold for $\sigma$, and $\predicate(\sigma):=\false$ otherwise.
Then, if $\predicate(\alpha)$ then $\predicate(\gamma)$.
\end{lemma}
\begin{proof}
The proof is a case consideration: which of the two cases holds in the definition of $\predicate$ hold for $\sigma$, for $\alpha$, and whether the distinguished $\loadkind$ and $\commitkind$ events are located in the same part $\alpha_j$.
We consider two of the cases. The other are similar.

Assume requirements (1)--(2) from Lemma~\ref{Lemma:EarlyReadLooksLikeThis} hold for $\alpha$ and requirements (1)--(3) from Lemma~\ref{Lemma:LoadFromMemoryLooksLikeThis} holds for sequentially consistent computation $\beta$.
If $\loadkind$ and $\commitkind$ events are in the same part, then
$\alpha=\alpha_1\cdot \deleted_1\cdots(\alpha_j'\cdot b\cdot\alpha_j''\cdot c\cdot d\cdot\alpha_j''')\cdot \deleted_j\cdots\alpha_n$,
$\beta=\beta_1\cdots c\cdot d\cdot\beta_2\cdot b\cdot\beta_3$, where
$b=(\loadkind,\tid,\anindex,\anaddr)$, $c=(\commitkind,\tid,\anindex')$, $d=(\propagatekind,\tid,\tid,\anindex',\anaddr)$, $\anindex'<\anindex$.
Consequently, $\gamma
=\projectionOf{\beta}{\alpha_1}\cdot \deleted_1\cdots\projectionOf{\beta}{\alpha_j}\cdot \deleted_j\cdots\projectionOf{\beta}{\alpha_n}
=\projectionOf{\beta}{\alpha_1}\cdot \deleted_1\cdots(\projectionOf{\beta_1}{\alpha_j}\cdot d\cdot\projectionOf{\beta_2}{\alpha_j}\cdot b\cdot\projectionOf{\beta_3}{\alpha_j})\cdot \deleted_j\cdots\projectionOf{\beta}{\alpha_n}$ --- looks like a read from memory situation.
We check requirements (1)--(3) of Lemma~\ref{Lemma:LoadFromMemoryLooksLikeThis} then.
First, $\projectionOf{\beta_2}{\alpha_j}$ must have no $\propagatekind$ events to thread $\tid$ with the address $\anaddr$ --- holds as $\beta_2$ does not have them.
Second, $\projectionOf{\beta_3}{\alpha_j}$ must have no commits of earlier writes in thread $\tid$ --- holds as $\beta_3$ does not have them.
Third, $\projectionOf{\beta}{\alpha_l}=\projectionOf{(\beta_1\cdot\beta_2\cdot\beta_3)}{\alpha_l}$, $l\in\intrange{i+1}{n}$ must have no commit events for stores with indices $\intrange{1}{i-1}$, the same address and thread id.
Consider $\projectionOf{\beta_1}{\alpha_l}$ --- if it has such an event $e$, then two stores to the same address, $e$ and $c$, are committed in different order in $\alpha'$ and $\beta$, which is impossible due to $\traceOf{\alpha'}=\traceOf{\beta}$.
Consider $\projectionOf{\beta_2}{\alpha_l}$ --- it does not have such an event, because $\beta_2$ does not have $\propagatekind$ events to address $\anaddr$, therefore, it does not have commits of own stores there too.
Consider $\projectionOf{\beta_3}{\alpha_l}$ --- it does not have such an event, because $\beta_3$ does not.
Finally, none of $\deleted_l$ events, $l\in\intrange{i+1}{n-1}$, must be a commit of earlier writes in thread $\tid$ --- holds, as these events belong to the last fetched instruction of a thread.

Consider the case when $\loadkind$ and $\commitkind$ events are in different parts, i.e.
$\alpha=\alpha_1\cdot \deleted_1\cdots(\alpha_j'\cdot b\cdot \alpha_j'')\cdots(\alpha_k'\cdot c\cdot d\cdot\alpha_k'')\cdots\alpha_n$,
$\beta=\beta_1\cdot c\cdot d\cdot \beta_2\cdot b\cdot\beta_3$, where $b,c,d$ are defined as before and $\anindex'<\anindex$.
Then, $\gamma
=\projectionOf{\beta}{\alpha_1}\cdot \deleted_1\cdots\projectionOf{\beta}{\alpha_j}\cdot \deleted_j\cdots\projectionOf{\beta}{\alpha_k}\cdots\projectionOf{\beta}{\alpha_n}
=\projectionOf{\beta}{\alpha_1}\cdot \deleted_1\cdots\projectionOf{\beta_1}{\alpha_j}\cdot\projectionOf{\beta_2}{\alpha_j}\cdot b\cdot\projectionOf{\beta_3}{\alpha_j}\cdot \deleted_j\cdots\projectionOf{\beta_1}{\alpha_k}\cdot c\cdot d\cdot\projectionOf{\beta_2}{\alpha_k}\cdot\projectionOf{\beta_3}{\alpha_k}\cdot \deleted_k\cdots\projectionOf{\beta}{\alpha_n}$ --- looks like an early read case.
Therefore, one must check that $\projectionOf{\beta_2}{\alpha_k}\cdot\projectionOf{\beta_3}{\alpha_k}\cdot \deleted_k\cdots\projectionOf{\beta}{\alpha_n}$ has no $\commitkind$ events matching $(\commitkind,\tid,\intrange{\anindex'+1}{\anindex-1},*,\anaddr)$.
Consider $\projectionOf{\beta_2}{\alpha_k}$ --- does not have such events, because they would be immediately followed by a $\propagatekind$ event to thread $\tid$ and address $\anaddr$, which contradicts requirement (2) of Lemma~\ref{Lemma:LoadFromMemoryLooksLikeThis}.
Consider $\projectionOf{\beta_3}{\alpha_k}$ --- does not have such events, because $\beta_3$ does not have them by requirement (3) of Lemma~\ref{Lemma:LoadFromMemoryLooksLikeThis}.
Consider $\projectionOf{\beta}{\alpha_l}$, $l\in\intrange{j+1}{n}$ --- does not have such events, because $\alpha_l$ do not have them by requirement (2) of Lemma~\ref{Lemma:EarlyReadLooksLikeThis}.
Finally, $\deleted_l$, $l\in\intrange{j+1}{n-1}$ belong to the last fetched instruction of a thread, therefore do not contain the described $\commitkind$ events.
\qed
\end{proof}

\begin{lemma}\label{Lemma:ReshuffledComputationIsFeasible}
$\gamma\in\computationsOf{\program}{\power}$.
\end{lemma}
\begin{proof}
We proceed by induction.
Assume
(1) $\gamma=\gamma_1\cdot\event\cdot\gamma_2$,
(2) $\initialpowerstate\transitionto{\gamma_1}\powerstate$,
and (3) all loads satisfied in $\gamma_1$ have read from the same stores as in $\alpha$.
We show that $\initialpowerstate\transitionto{\gamma_1\cdot\event}\powerstate'$ and all loads satisfied in $\gamma_1\cdot\event$ have read from the same stores as in $\alpha$.
Let $\powerstate=(\rtstatesmap,\ststate)$ and $\rtstatesmap(\tid)=(\fetched,\committed,\loaded)$.
Consider the event $\event$.
\begin{description}

\item[$(\fetchkind,\tid,\anindex)$]
A transition labeled by $\event$ from state $\powerstate$ is feasible due to Lemma~\ref{Lemma:ProgramOrderIsTheSame} and the fact that feasibility of a fetch transition is conditioned solely on the previous $\fetchkind$ transition with the same thread id.

\item[$(\loadkind,\tid,\anindex,\anaddr)$]
For the transition to be feasible, $\getaddr(\anindex)=\anaddr$ must hold.
In order to have $\getaddr(\tid,\anindex)\neq\bot$, all loads in thread $\tid$, on which $\getaddr(\tid,\anindex)$ depends, must be satisfied.
Note that these loads are the same in $\alpha$ and $\gamma$ due to Lemma~\ref{Lemma:ProgramOrderIsTheSame}.
Since $\alpha\in\computationsOf{\program}{\power}$, these $\loadkind$ events occurred before $\event$ in $\alpha$.
Let $\event'$ be one of these $\loadkind$ events.
If $\event'\in\alpha_i$ and $\event\in\alpha_j$, $i<j$, or $\event'\in\setcond{\deleted_i}{i\in\intrange{1}{n-1}}$, or $\event\in\setcond{\deleted_i}{i\in\intrange{1}{n-1}}$, then $\event'$ and $\event$ are located in $\gamma$ in the same order.
If $\event',\event\in\alpha_i$, then $\event',\event\in\beta$.
Since the $\programorder$ components of $\traceOf{\alpha}$ and $\traceOf{\beta}$ match up to a single deleted arc, $\event'$ and $\event$ are located in $\beta$ (therefore, in $\projectionOf{\beta}{\alpha_i}$ and $\gamma$) in this order.
By inductive assumption~(3) and the fact that functions in $\functiondomain$ are deterministic, $\getaddr(\tid,\anindex)=\anaddr$ holds.

Assume the load $(\tid,\anindex)$ has read from a store $(\tid',\anindex')$ in $\alpha$.
Then, by Lemmas~\ref{Lemma:EarlyReadLooksLikeThis}, \ref{Lemma:LoadFromMemoryLooksLikeThis}, \ref{Lemma:SourceRelationIsTheSame}, either conditions (1)--(3) of Lemma~\ref{Lemma:LoadFromMemoryLooksLikeThis} hold, or conditions (1)--(2) of Lemma~\ref{Lemma:EarlyReadLooksLikeThis} hold.
In the former case, $(\propagatekind,\tid,\tid',\anindex',\anaddr)$ is the last $\propagatekind$ event to $\tid$ with address $\anaddr$, therefore, a load from memory transition reading $(\tid',\anindex')$ is feasible from state $\powerstate$.
In the latter case, $(\tid',\anindex')$ is the latest non-committed store to address $\anaddr$, and an early read transition reading $(\tid',\anindex')$ is possible.
The proof that $\getaddr(\tid',\anindex')\neq\bot$ is similar to the proof that $\getaddr(\tid,\anindex)\neq\bot$.

\item[$(\commitkind,\tid,\anindex)$]
The proof of $\getaddr(\tid,\anindex)\neq\bot$ and $\getvalue(\tid,\anindex)\neq\bot$ is similar to the proof of $\getaddr(\tid,\anindex)\neq\bot$ in the previous case.
If $\fetched[\anindex]$ is a load or a store, there must be no preceding loads and stores to unknown addresses, which holds and can be proven in a similar way.
If $\fetched[\anindex]$ is a load, requirement $\loaded[\anindex]\neq\bot$ holds for the same reasons.
If $\fetched[\anindex]$ is a conditional, requirement $\getvalue(\tid,\anindex)\neq 0$ holds by inductive assumption (3), the fact that functions in $\functiondomain$ are deterministic, and the fact that $\alpha\in\computationsOf{\program}{\power}$.

\item[$(\commitkind,\tid,\anindex,\coherencekey,\anaddr)$]
Value $\coherencekey$ is unique, since it was unique in $\alpha$, and $\alpha$ and $\gamma$ consist of the same commit events.
We check $\coherence(\propagated(\tid,\anaddr))<\coherencekey$.
Assume it does not hold.
Then, there is $\event'=(\propagatekind,\tid,\tid',\anindex',\anaddr)$, where $\coherence(\tid',\anindex')>\coherencekey$, and $\event'$, $\event$ are located in $\gamma$ in this order.
If $\event'\in\alpha_i$, $\event\in\alpha_j$, $i<j$, or $\event'\in\setcond{\deleted_i}{i\in\intrange{1}{n-1}}$, or $\event\in\setcond{\deleted_i}{i\in\intrange{1}{n-1}}$, these events are located in $\alpha$ in this order, which contradicts $\alpha\in\computationsOf{\program}{\power}$.
If $\event',\event\in\alpha_i$, these events are located in $\beta$ in this order, which contradicts $\beta\in\computationsOf{\program}{\power}$.

This transition is immediately followed by a $\propagatekind$ transition in $\gamma$, since it did so in $\alpha$ and $\beta$ (unless $e\in\setcond{\deleted_i}{i\in\intrange{1}{n-1}}$, which is a simpler case), and by properties of projection.

\item[$(\propagatekind,\tid,\tid',\anindex',\anaddr)$]
The requirement $\coherence(\propagated(\tid,\anaddr))<\coherence(\tid',\anindex')$ is proven similarly to $\coherence(\propagated(\tid,\anaddr))<\coherencekey$ in the previous case.

\end{description}

As shown above, $\initialpowerstate\transitionto{\gamma}\powerstate$.
What is left to check, is that $\powerstate\in\finalpowerstates$.
The requirement that all fetched instructions are committed trivially holds: $\beta$ includes the same commit events as $\alpha'$, therefore, by definition, $\gamma$ contains the same commit events as $\alpha$.
The other two requirements that loads and stores agree with the coherence order hold due to Lemma~\ref{Lemma:ProgramOrderIsTheSame}, the inductive assumption (3), and the fact that $\alpha$ and $\gamma$ consist of the same commit events (i.e.\ the coherence keys of matching stores are equal in these computations).
\qed
\end{proof}

\begin{lemma}\label{Lemma:ReshuffledComputationHasTheSameTrace}
$\traceOf{\gamma}=\traceOf{\alpha}$
\end{lemma}
\begin{proof}
Equality of $\programorder$ follows from Lemma~\ref{Lemma:ProgramOrderIsTheSame}.
Equality of source relation follows from Lemmas~\ref{Lemma:EarlyReadLooksLikeThis}, \ref{Lemma:LoadFromMemoryLooksLikeThis}, \ref{Lemma:SourceRelationIsTheSame}, \ref{Lemma:ReshuffledComputationIsFeasible}.
Store order is determined by $\anaddr$ and $\coherencekey$ components of store commit events.
Since computations $\alpha$ and $\gamma$ consist of the same commit events, the $\coherenceorder$ relations in the traces of $\alpha$ and $\gamma$ are the same.
\qed
\end{proof}

\begin{lemma}\label{Lemma:ReshuffledComputationIsFeasibleAndHasTheSameTrace}
$\gamma\in\computationsOf{\program}{\power}$ and $\traceOf{\gamma}=\traceOf{\alpha}$.
\end{lemma}
\begin{proof}
Corollary of Lemmas~\ref{Lemma:ReshuffledComputationIsFeasible} and~\ref{Lemma:ReshuffledComputationHasTheSameTrace}.
\qed
\end{proof}

Without loss of generality we may assume that all $\fetchkind$ events of $\alpha$ are located within $\alpha_1\cdot \deleted_1$: every thread can always first fetch all instructions and in the rest of the computation only execute them; such a reordering does not change the trace.
Also, note that the maximal number of events an instruction can generate is $\cardinalityOf{\program}+2$.
This bound is achieved by a store that is fetched, committed, and propagated to all threads.
Then the following lemma holds:
\begin{lemma}\label{Lemma:ReshuffledComputationIsInNormalForm}
Computation $\gamma$ is in normal form of degree $\cardinalityOf{\program}+3$.
\end{lemma}
\begin{proof}
By definition of $\gamma$ and properties of projection.
\qed
\end{proof}
Together with Lemma~\ref{Lemma:ShashaSnir} this proves Theorem~\ref{Theorem:NormalFormComputationsAreEnough}.

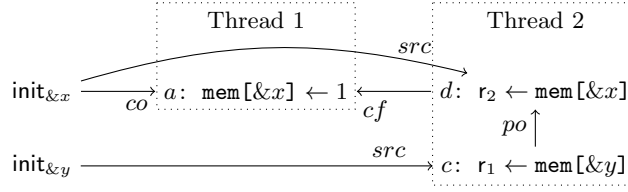
\begin{figure}[t]
\centering
\begin{tikzpicture}[nodes={rectangle,draw=none,fill=none}]
  \matrix[row sep=0.5cm,column sep=1cm,nodes={rectangle,draw=none,fill=none}] {
    & \node (thread1) {Thread 1}; & \node (thread2) {Thread 2}; \\
    \node (initialstorex) {$\initialstore_{\&x}$}; & \node (storex) {$a\colon$ \thestore{\&x}{1}}; & \node (loadx) {$d\colon$ \theload{\areg_2}{\&x}}; \\
    \node(initialstorey) {$\initialstore_{\&y}$}; &  & \node (loady) {$c\colon$ \theload{\areg_1}{\&y}}; \\
  };
  \draw[->] (loady) edge node[midway,left] {$po$} (loadx);
  \draw[->] (loadx) edge node[near end,below] {$cf$} (storex);
  \draw[->] (initialstorex) edge[bend left=15] node[very near end,above] {$src$} (loadx);
  \draw[->] (initialstorex) edge node[near end,below] {$co$} (storex);
  \draw[->] (initialstorey) edge node[very near end,above] {$src$} (loady);

  \draw[dotted] (thread1.north -| storey.west) rectangle (storex.south -| storey.east);
  \draw[dotted] (thread2.north -| loady.west) rectangle (loady.south -| loady.east);
\end{tikzpicture}

\caption{Trace of the computations $\alpha'$ and $\beta$ from Example~\ref{Example:NormalFormComputation}.}
\label{Figure:TraceMPShortened}
\end{figure}

\begin{example}\label{Example:NormalFormComputation}
Consider $\alpha:=\fetchkind(c)\cdot\fetchkind(d)\cdot\fetchkind(a)\cdot\cancel{\fetchkind(b)}\cdot\commitkind(a)\cdot\propagatekind(a,1)\cdot\cancel{\commitkind(b)}\cdot\cancel{\propagatekind(b,1)}\cdot\cancel{\propagatekind(b,2)}\cdot\loadkind(c)\cdot\loadkind(d)\cdot\commitkind(d)\cdot\commitkind(c)$, which is essentially $\sigma_{\mathit{MP}}$ with $\fetchkind$ events moved to the front.
We cancel the $\deleted_i$ events (crossed out) related to the store instruction $b$, as $b$ is the last instruction of thread 1 and no address depends on it (we could also cancel the events of $d$ instead).
Therefore,
$\alpha_1:=\fetchkind(c)\cdot\fetchkind(d)\cdot\fetchkind(a)$,
$\alpha_2:=\commitkind(a)\cdot\propagatekind(a,1)$,
$\alpha_3:=\alpha_4:=\emptysequence$,
$\alpha_5:=\loadkind(c)\cdot\loadkind(d)\cdot\commitkind(d)\cdot\commitkind(c)$,
and $\alpha':=\alpha_1\cdot\alpha_2\cdot\alpha_3\cdot\alpha_4\cdot\alpha_5$.
The trace of $\alpha'$ is shown in Figure~\ref{Figure:TraceMPShortened}.
The SC computation with the same trace is $\beta:=\fetchkind(c)\cdot\loadkind(c)\cdot\commitkind(c)\cdot\fetchkind(d)\cdot\loadkind(d)\cdot\commitkind(d)\cdot\fetchkind(a)\cdot\commitkind(a)\cdot\propagatekind(a,1)\cdot\propagatekind(a,2)$.
The normal-form computation is $\gamma:=\projectionOf{\beta}{\alpha_1}\cdot \deleted_1\cdots\projectionOf{\beta}{\alpha_5}=(\fetchkind(c)\cdot\fetchkind(d)\cdot\fetchkind(a))\cdot\fetchkind(b)\cdot(\commitkind(a)\cdot\propagatekind(a,1))\cdot\commitkind(b)\cdot\propagatekind(b,1)\cdot\propagatekind(b,2)\cdot(\loadkind(c)\cdot\commitkind(c)\cdot\loadkind(d)\cdot\commitkind(d))$.
It is feasible and has the same trace as $\alpha$ and $\sigma_{\mathit{MP}}$ (Figure~\ref{Figure:TraceMP}).
\end{example}

\section{From Normal-Form Computations to Emptiness}\label{Section:Reduction}

We now reduce robustness to language emptiness.
First, we define a multiheaded automaton capable of generating all normal-form computations of a program.
Next, we intersect it with regular languages that check cyclicity of the happens-before relation.
Altogether, the program is robust iff the intersection is empty.

\subsection{Generating Normal-Form Computations}\label{Section:Reduction:Generating}
To generate all normal-form computations, we use so-called multiheaded automata~\cite{calin2013}. 
Essentially, a multiheaded automaton generates a computation $\sigma_1\ldots \sigma_n$ by simultaneously generating its parts $\sigma_i$.
The automaton has a head for each part, and the transition relation defines the head producing an event.
Formally, an \emph{$n$-headed automaton over $\Sigma$} is an automaton operating on the extended alphabet $\intrange{1}{n}\times\alphabet$: $\automaton=(\states,\intrange{1}{n}\times\alphabet,\transitions,\initialstate,\finalstates)$.
The \emph{language of $\automaton$} is $\langOf{\automaton}:=\setcond{\secondOf{\projectionOf{\sigma}{(\set{1}\times\alphabet)}\cdots\projectionOf{\sigma}{(\set{n}\times\alphabet)}}}{\initialstate\transitionto{\sigma}\state\text{ for some }\state\in\finalstates}$, where $\secondOf{(a_1,b_1)\cdots(a_m,b_m)}:=b_1\cdots{}b_m$.
Multiheaded automata are closed under intersection with regular languages.
Moreover, for finite multiheaded automata language emptiness is $\nlogspace$-complete~\cite{calin2013}:

\begin{lemma}[\cite{calin2013}]\label{Lemma:MultiheadedIntersection}
Consider an $n$-headed automaton $U$ and a finite automaton $V$ over a common alphabet $\alphabet$.
There is an $n$-headed automaton $W$ with $\langOf{W}=\langOf{U}\cap\langOf{V}$ with the number of states $\cardinalityOf{\statesOf{W}}\leq\cardinalityOf{\statesOf{U}}\cdot\cardinalityOf{\statesOf{V}}^{2n}+1$.
\end{lemma}

\begin{lemma}[\cite{calin2013}]\label{Lemma:MultiheadedEmptinessComplexity}
Emptiness for $n$-headed automata is $\nlogspace$-complete.
\end{lemma}

We will generate all normal-form computations of program $\program$ with the $n$-headed automaton $\mhautomaton(\program):=(\mhstates,\events,\mhtransitions,\initialmhstate,\finalmhstates)$, where $n:=\cardinalityOf{\program}+3$.
The automaton generates all events related to a single instruction in one shot, but, possibly, in different parts of the computation.
All $\fetchkind$ events are generated in the first part of the computation.
In order to generate them, the automaton keeps track of the destination state of the last fetched instruction in each thread (component $\mcontrolstate$ of the automaton state).

Each instruction can only read the last value written to a register.
Therefore, the automaton only needs to remember $\cardinalityOf{\regdomain}$ register values per thread (component $\mregvalue$).
However, an instruction cannot be executed until the values of all registers that it reads become known.
To obey this restriction, the automaton memorizes the part of the computation in which the register value gets computed ($\mregcomphead$).
For example, while handling an assignment $\theassign{\areg_1}{\areg_1+\areg_2}$, the automaton learns that the new value of $\areg_1$ is the sum of the current values of $\areg_1$ and $\areg_2$.
It also remembers that this value is available no earlier than the current values of $\areg_1$ and $\areg_2$ are computed.
Similarly, the automaton remembers the parts of the computation in which the addresses of load and store instructions become known ($\maddrcomphead$), and certain kinds of instructions get committed ($\mregcommhead$, $\massumecommhead$, $\maddrcommhead$).

The automaton has to keep a separate memory state for each thread and for each part of the computation.
The memory state of a thread in a part is updated when a store instruction gets propagated to this thread in this part.
When a load instruction is handled, the automaton chooses a part where the $\loadkind$ event takes place and uses the memory state of that part.
Besides the memory valuation ($\mmemvalue$), the memory state includes coherence keys ($\mlastkey$) to guarantee that the generated computation 
respects the coherence order.

When starting the computation, 
the automaton non-deterministically guesses the memory valuations and coherence keys for all parts of the computation (except the first one).
Upon termination, the automaton checks that the parts of the computation generated by each head fit together at the concatenation points.
This ensures the overall computation is valid for the program. 
The trick is to remember the guess of the initial memory valuations and coherence keys in immutable components of the automaton state ($\guessed{\mmemvalue}$, $\guessed{\mlastkey}$).
The final states require that the current memory state in part $\headindex$ of the computation coincides with the guessed initial state in part $\headindex+1$.

\subsubsection{State space}\label{Section:Reduction:StateSpace}

A state from $\mhstates$ (except the special initial state $\initialmhstate$) includes the following information:
\begin{itemize}
\item $\mcontrolstate(\tid)$ gives the current control state of thread $\tid$.

\item $\mregcomphead(\tid,\areg)$ gives the part in which last value assigned to register $\areg$ in thread $\tid$ gets computed.
\item $\mregvalue(\tid,\areg)$ gives this computed value.
\item $\mregcommhead(\tid,\areg)$ gives the part in which the last instruction assigning a value to register $\areg$ in thread $\tid$ gets committed.

\item $\massumecommhead(\tid)$ gives the part in which the latest fetched condition in thread $\tid$ is committed.

\item $\mmemvalue(\tid,\anaddr,\headindex)$ gives the value of the last write to $\anaddr$ propagated to thread $\tid$ in the part $\headindex$ or earlier.
\item $\mlastkey(\tid,\anaddr,\headindex)$ gives the coherence key of the last write to $\anaddr$ propagated to thread $\tid$ in the part $\headindex$ or earlier.

\item $\guessed{\mmemvalue}$, $\guessed{\mlastkey}$ are immutable copies of the guessed values of the previous two components (see \descref{MH-GUESS} below).

\item $\mearlymemvalue(\tid,\anaddr,\headindex)$ gives the value written by the last fetched store to $\anaddr$ which is still in-flight in the part $\headindex$ of computation,
$\bot$ if there is no such store,
$\top$ if the value of the store is unknown or there is a later in-flight store in this part with an unknown address.

\item $\maddrcomphead(\tid)$ gives the leftmost part of the computation, in which the addresses of all already fetched memory accesses are computed.
\item $\maddrcommhead(\tid,\anaddr)$ gives the rightmost part of the computation having a commit to address $\anaddr$ by thread $\tid$.

\item $\minstrcount(\tid)$ gives the number of instructions fetched in thread $\tid$.
\end{itemize}
The initial state $\initialmhstate$ does not contain any information.

\subsubsection{Transition relation}\label{Section:Reduction:TransitionRelation}

We define transitions by specifying the new (primed) values of the state components and the label $\lambda$ of the transition.
First, we define the transition guessing the initial memory state in each part of the computation:

\begin{description}
\descitem{MH-GUESS}
Assume the current state is $\initialmhstate$.
Then, there are transitions to the states satisfying
$\mcontrolstate':=\lambda\tid.\initialcontrolstate_{\tid}$,
$\mregcomphead':=\lambda\tid.\lambda\areg.1$,
$\mregvalue':=\lambda\tid.\lambda\areg.0$,
$\mregcommhead':=\lambda\tid.\lambda\areg.1$,
$\massumecommhead':=\lambda\tid.1$,
$\mearlymemvalue':=\lambda\tid.\lambda\anaddr.\lambda\headindex.\bot$,
$\mmemvalue'=\guessed{\mmemvalue}'$,
$\mlastkey'=\guessed{\mlastkey}'$,
$\maddrcomphead':=\lambda\tid.1$,
$\maddrcommhead':=\lambda\tid.\lambda\anaddr.1$,
$\minstrcount':=\lambda\tid.0$.
Also, $\mmemvalue'(\tid,\anaddr,1):=0$, $\mlastkey'(\tid,\anaddr,1):=0$ for all $\tid\in\tiddomain$, $\anaddr\in\addrdomain$.
Moreover, $\mlastkey'(\tid,\anaddr,\headindex)\leq\mlastkey'(\tid,\anaddr,\headindex+1)$ for $\headindex\in\intrange{1}{n-1}$, $\tid\in\tiddomain$, $\anaddr\in\addrdomain$ (we assume $\mlastkey'(\tid,\anaddr,n):=\infty$).
$\lambda:=\emptysequence$.
\end{description}

Fix a state $\mhstate$.
We overload $\evaluate(\tid,\anexpr)$ to mean the value of expression $\anexpr$ for the valuation of registers defined by $\lambda\areg.\mregvalue(\tid,\areg)$.

Let $\headdomain:=\intrange{1}{n}$.
Let $\tid\in\tiddomain$, $\mcontrolstate(\tid)=\controlstate_1$, $\instruction=\controlstate_1\transitionto{\command}\controlstate_2\in\instructions_{\tid}$.
Let $\headindex_1:=1$.
Let $\headindex_2\in\headdomain$, $\headindex_2\geq\headindex_1$, $\headindex_2\geq\mregcomphead(\tid,\areg)$ for each register $\areg$ read in $\command$.
Let $\headindex_3\in\headdomain$, $\headindex_3\geq\headindex_2$, $\headindex_3\geq\mregcommhead(\tid,\areg)$ for each register $\areg$ read in $\command$, $\headindex_3\geq\massumecommhead(\tid)$.
Let $\anindex:=\minstrcount(\tid)+1$ and $\minstrcount':=\minstrcount[\tid\hookleftarrow\anindex]$.
Depending on the type of $\command$, there are the following transitions from $\mhstate$ labeled by events $\lambda$:
\begin{description}
\descitem{MH-ASSIGN}
$\command=\theassign{\areg}{\anexpr_{\aval}}$.
Let $\aval:=\evaluate(\tid,\anexpr_{\aval})$.
Then $\mregvalue':=\mregvalue[(\tid,\areg)\hookleftarrow\aval]$, $\mregcomphead':=\mregcomphead[(\tid,\areg)\hookleftarrow\headindex_2]$, $\mregcommhead':=\mregcommhead[(\tid,\areg)\hookleftarrow\headindex_3]$.
$\lambda:=(\headindex_1,\fetchkind,\tid,\instruction)\cdot(\headindex_3,\commitkind,\tid,\anindex)$.

\descitem{MH-ASSUME}
$\command=\theassume{\anexpr_{\aval}}$.
Let $\evaluate(\tid,\anexpr_{\aval})\neq 0$.
Then $\massumecommhead':=\massumecommhead[\tid\hookleftarrow\headindex_3]$.
$\lambda:=(\headindex_1,\fetchkind,\tid,\instruction)\cdot(\headindex_3,\commitkind,\tid,\anindex)$.

\descitem{MH-LOAD}
$\command=\theload{\areg}{\anexpr_{\anaddr}}$.
Let $\anaddr:=\evaluate(\tid,\anexpr_{\anaddr})$.
Let $\headindex_3\geq\maddrcommhead(\tid,\anaddr)$.
If $\mearlymemvalue(\tid,\anaddr)=\bot$, let $\aval:=\mmemvalue(\tid,\anaddr,\headindex_2)$ (load from memory case).
Otherwise, let $\aval:=\mearlymemvalue(\tid,\anaddr,\headindex_2)$ and assume $\aval\neq\top$ (early read case).
Then $\mregvalue':=\mregvalue[(\tid,\areg)\hookleftarrow\aval]$, $\mregcomphead':=\mregcomphead[(\tid,\areg)\hookleftarrow\headindex_2]$, $\mregcommhead':=\mregcommhead[(\tid,\areg)\hookleftarrow\headindex_3]$, $\maddrcomphead':=\maddrcomphead[\tid\hookleftarrow\max\set{\maddrcomphead(\tid),\headindex_2}]$, $\maddrcommhead':=\maddrcommhead[(\tid,\anaddr)\hookleftarrow\headindex_3]$.
$\lambda:=(\headindex_1,\fetchkind,\tid,\instruction)\cdot(\headindex_2,\loadkind,\tid,\anindex,\anaddr)\cdot(\headindex_3,\commitkind,\tid,\anindex)$.

\descitem{MH-STORE}
$\command=\thestore{\anexpr_{\anaddr}}{\anexpr_{\aval}}$.
Let $\anaddr:=\evaluate(\tid,\anexpr_{\anaddr})$.
Assume $\headindex_3\geq\maddrcomphead(\tid)$, $\headindex_3\geq\maddrcommhead(\tid,\anaddr)$.
Let $\aval:=\evaluate(\tid,\anexpr_{\aval})$.
Let $\coherencekey\in\rationalnumbers$, $\coherencekey\neq\mlastkey(\tid,\anaddr,\headindex)$ for any $\tid\in\tiddomain$, $\anaddr\in\addrdomain$, $\headindex\in\headdomain$.
Then $\mearlymemvalue':=\mearlymemvalue[(\tid,\anaddr,\intrange{\headindex_1}{\headindex_2-1})\hookleftarrow\top),(\tid,\anaddr,\intrange{\headindex_2}{\headindex_3-1})\hookleftarrow\aval]$.
We also set $\mearlymemvalue':=\mearlymemvalue'[(\tid,\anaddr',\headindex)\hookleftarrow\top]$ for all $\anaddr'\in\addrdomain\setminus\set{\anaddr}$, $\headindex\in\intrange{\headindex_1}{\headindex_2-1}$ with $\mearlymemvalue(\tid,\anaddr',\headindex)\in\datadomain$.
We define $\maddrcomphead':=\maddrcomphead[\tid\hookleftarrow\max\set{\maddrcomphead(\tid),\headindex_2}]$, $\maddrcommhead':=\maddrcommhead[(\tid,\anaddr)\hookleftarrow\headindex_3]$.
Let $T\subseteq\tiddomain\setminus\set{\tid}$, initially $\mmemvalue':=\mmemvalue$, $\mlastkey':=\mlastkey$, and $\lambda:=(\headindex_1,\fetchkind,\tid,\instruction)\cdot(\headindex_3,\commitkind,\tid,\anindex,\coherencekey,\anaddr)$.
For $\tid'=\tid$ and for each $\tid'\in T$: let $\headindex\in\headdomain$, $\headindex\geq\headindex_3$ ($\headindex:=\headindex_3$ for $\tid'=\tid$), $\mlastkey(\tid',\anaddr,\headindex)<\coherencekey\leq\guessed{\mlastkey}(\tid',\anaddr,\headindex+1)$, then $\mmemvalue':=\mmemvalue'[(\tid',\anaddr,\headindex)\hookleftarrow\aval]$, $\mlastkey':=\mlastkey'[(\tid',\anaddr,\headindex)\hookleftarrow\coherencekey]$, $\lambda:=\lambda\cdot(\headindex,\propagatekind,\tid',\tid,\anindex,\anaddr)$.

\end{description}
For brevity we allowed a single transition to be labeled by several events.
An automaton with such transitions can be trivially translated to the canonical form by breaking one such transition into several consecutive ones.

\subsubsection{Final states}\label{Section:Reduction:FinalStates}
The set of final states $\finalmhstates$ is a subset of $\mhstates\setminus\set{\initialmhstate}$ consisting of all states with
$\mmemvalue(\tid,\anaddr,\headindex)=\guessed{\mmemvalue}(\tid,\anaddr,\headindex+1)$,
$\mlastkey(\tid,\anaddr,\headindex)=\guessed{\mlastkey}(\tid,\anaddr,\headindex+1)$
for all $\tid\in\tiddomain$, $\anaddr\in\addrdomain$, $\headindex\in\intrange{1}{n-1}$.

\subsubsection{Soundness and completeness}

\begin{lemma}\label{Lemma:GeneratesOnlyCorrectComputations}
$\langOf{\mhautomaton}\subseteq\computationsOf{\program}{\power}$.
\end{lemma}
\begin{proof}

Consider $\sigma=\lambda_1\cdots\lambda_m$, such that $\initialmhstate\transitionto{\lambda_1}\mhstate_1\transitionto{\lambda_2}\cdots\transitionto{\lambda_m}\mhstate_m\in\finalmhstates$.
For $\headindex\in\headdomain$, let $\tau_{\headindex}^s:=\secondOf{\projectionOf{(\lambda_1\cdots\lambda_s)}{(\set{\headindex}\times\events})}$, $s\in\intrange{0}{m}$.

Let $(\powerstate_1^0\ldots\powerstate_n^0)\in(\powerstates)^{n}$ be the states of $\powerautomaton$ defined so that \descref{SND-B} holds for $s=0$ (see below).
By induction on $s\in\intrange{1}{m}$ we show:
\begin{description}
\descitem{SND-A} $\powerstate_{\headindex}^0\transitionto{\tau_{\headindex}^s}\powerstate_{\headindex}^s$.
\descitem{SND-B} For all $\tid\in\tiddomain$, $\headindex\in\headdomain$, $\powerstate_{\headindex}^s=(\rtstatesmap,(\coherence,\propagated))$, $\rtstatesmap(\tid)=(\fetched,\committed,\loaded)$ holds:
  \begin{description}
  \descitem{SND-B1} $\fetched$ is the list of instructions fetched by $\projectionOf{\projectionOf{(\tau_1^m\cdots\tau_{\headindex-1}^m\cdot\tau_{\headindex}^s)}{\fetchkind}}{\tid}$.
  \descitem{SND-B2} $\committed$ consists of the indices of instructions committed by $\projectionOf{(\projectionOf{\tau_1^m\cdots\tau_{\headindex-1}^m\cdot\tau_{\headindex}^s)}{\commitkind}}{\tid}$.
  \descitem{SND-B3} $\loaded$ contains the information about the stores being read by loads in $(\tau_1^m\cdots\tau_{\headindex-1}^m\cdot\tau_{\headindex}^s)$ determined according to Lemmas~\ref{Lemma:EarlyReadLooksLikeThis} and~\ref{Lemma:LoadFromMemoryLooksLikeThis}.
  \descitem{SND-B4} $\coherence(\tid,\anindex)=\coherencekey$ if $(\commitkind,\tid,\anindex,\coherencekey,\anaddr)\in\tau_1^m\cdots\tau_{\headindex-1}^m\cdot\tau_{\headindex}^s$ for some $\anaddr\in\addrdomain$, otherwise, $\coherence(\tid,\anindex)=\bot$.
  \descitem{SND-B5} $\propagated(\tid,\anaddr)=(\tid',\anindex')$ if $(\propagatekind,\tid,\tid',\anindex',\anaddr)=\lastOf{\projectionOf{(\tau_1^m\cdots\tau_{\headindex-1}^m\cdot\tau_{\headindex}^s)}{(\propagatekind,\tid,*,*,\anaddr)}}$, otherwise, $\propagated(\tid,\anaddr)=\initialstore_{\anaddr}$.
  \end{description}
\descitem{SND-C} For each $\tid\in\tiddomain$: $\mcontrolstate(\tid)=\destinationstateOf{\lastOf{\powerstate_1^s.\rtstatesmap(\tid).\fetched}}$ (or $\initialcontrolstate_{\tid}$ if no instructions were fetched).
\descitem{SND-D} For each $\tid\in\tiddomain$, $\areg\in\regdomain$, for each $\headindex\in\intrange{\mregcomphead(\tid,\areg)}{n}$: $\mregvalue(\tid,\areg)=\evaluate(\tid,\minstrcount(\tid)+1,\areg)$ computed for the state $\powerstate_\headindex^s$.
\descitem{SND-E} For each $\tid\in\tiddomain$, $\areg\in\regdomain$, $\headindex\in\intrange{\mregcommhead(\tid,\areg)}{n}$: let $\anindex$ be the index of the latest instruction in $\powerstate_{\headindex}^s.\rtstatesmap(\tid).\fetched$ writing to $\areg$, then $\anindex\in\powerstate_{\headindex}^s.\rtstatesmap(\tid).\committed$.
\descitem{SND-F} For each $\tid\in\tiddomain$, $\headindex\in\intrange{\massumecommhead(\tid)}{n}$: $\powerstate_{\headindex}^s$ does not contain uncommitted conditional instructions in thread $\tid$ having indices $\leq\minstrcount(\tid)$.
\descitem{SND-G} For each $\tid\in\tiddomain$, $\anaddr\in\addrdomain$, $\headindex\in\headdomain$: let $w:=\powerstate_{\headindex}^s.\propagated(\tid,\anaddr)$.
If $w=\initialstore_{\anaddr}$, $\mmemvalue(\tid,\anaddr,\headindex)=0$.
If $w=(\tid',\anindex')$, $\mmemvalue(\tid,\anaddr,\headindex)=\getvalue(\tid',\anindex')$ computed in $\powerstate_{\headindex}^s$.
\descitem{SND-H} For each $\tid\in\tiddomain$, $\anaddr\in\addrdomain$, $\headindex\in\headdomain$: $\guessed{\mlastkey}(\tid,\anaddr,\headindex)\leq\powerstate_{\headindex}^s.\coherence(\powerstate_{\headindex}^s.\propagated(\tid,\anaddr))=\mlastkey(\tid,\anaddr,\headindex)\leq\guessed{\mlastkey}(\tid,\anaddr,\headindex+1)$.
\descitem{SND-K} For each $\tid\in\tiddomain$, $\anaddr\in\addrdomain$, $\headindex\in\headdomain$: let $\anindex\in\naturalnumbers$ be the maximal index, such that $\powerstate_{\headindex}^s.\rtstatesmap(\tid).\fetched[\anindex]$ is a store, $\getaddr(\tid,\anindex)=\anaddr$ in $\powerstate_n^s$.
Let $\anindex'$ be the maximal index, such that $\powerstate_{\headindex}^s.\rtstatesmap(\tid).\fetched[\anindex']$ is a store, $\getaddr(\tid,\anindex')\in\set{\bot,\anaddr}$ in $\powerstate_{\headindex}^s$.
Then $\mearlymemvalue(\tid,\anindex,\headindex)=\bot$ if such $\anindex$ does not exist or $\anindex\in\powerstate_{\headindex}^s.\rtstatesmap(\tid).\committed$.
Otherwise, $\mearlymemvalue(\tid,\anindex,\headindex)=\top$ if $\getaddr(\tid,\anindex')=\bot$ or $\getvalue(\tid,\anindex)=\bot$ in $\powerstate_{\headindex}^s$.
Otherwise, $\mearlymemvalue(\tid,\anindex,\headindex)=\getvalue(\tid,\anindex)$ computed in $\powerstate_{\headindex}^s$.
\descitem{SND-L} For each $\tid\in\tiddomain$, $\headindex\in\intrange{\maddrcomphead(\tid)}{n}$, $\anindex\in\intrange{1}{\lengthOf{\powerstate_{\headindex}^s.\rtstatesmap(\tid).\fetched}}$: $\getaddr(\tid,\anindex)\neq\bot$ in $\powerstate_{\headindex}^s$.
\descitem{SND-M} For each $\tid\in\tiddomain$, $\anaddr\in\addrdomain$, $\headindex\in\intrange{\maddrcommhead(\tid,\anaddr)}{n}$: if $\getaddr(\tid,\anindex)=\anaddr$ in $\powerstate_n^s$ for some $\anindex$, then $\anindex\in\powerstate_{\headindex}^s.\rtstatesmap(\tid).\committed$.
\end{description}

Finally we will show that $\powerstate_{\headindex}^{m}=\powerstate_{\headindex+1}^0$ for all $\headindex\in\intrange{1}{n-1}$ and $\powerstate_{n}^{m}\in\finalpowerstates$, thus proving the claim of the lemma.

Base case: $s=1$, we must show that there $\mhstate_1$ satisfies the inductive statement.
This is easy to check by definition of the destination state of \descref{MH-GUESS} transition.

Step case: assume the inductive statement holds for some $s\in\intrange{0}{m-1}$.
Consider $\lambda_s$ (for notational convenience and without loss of generality we assume below that $\headindex_j\neq\headindex_{j'}$ for $j\neq j'$):
\begin{description}

\item[Assignment]
$\lambda_s=(\headindex_1,\fetchkind,\tid,\instruction)\cdot(\headindex_3,\commitkind,\tid,\anindex)$, $\instruction=\controlstate_1\transitionto{\theassign{\areg}{\anexpr_{\aval}}}\controlstate_2$.
Let $\event_1:=(\fetchkind,\tid,\instruction)$, $\event_3:=(\commitkind,\tid,\anindex)$.

We need to show that $\powerstate_{\headindex_1}^{s-1}\transitionto{\event_1}\powerstate_{\headindex_1}^s$, i.e. that the assignment instruction can be fetched.
This follows from the choice of $\headindex_1:=1$ in \descref{MH-ASSIGN} and \descref{SND-B1}, \descref{SND-C}.

We also need to show that $\powerstate_{\headindex_3}^{s-1}\transitionto{\event_3}\powerstate_{\headindex_3}^s$, i.e. that the assignment instruction can be committed.
First, the $\event_3$ transition requires the instruction being committed to be fetched, which holds due to \descref{SND-B1} and $\headindex_3\geq\headindex_1$.
Second, this instruction must be not committed yet, which holds by \descref{SND-B2} and the fact that $\mhautomaton$ commits each instruction once and only once.
Third, all control dependencies must be committed.
This is by the choice of $\headindex_3$ in \descref{MH-ASSIGN} and \descref{SND-F}.
Fourth, all the preceding data dependencies must be committed.
This is by the choice of $\headindex_3$ in \descref{MH-ASSIGN} and \descref{SND-E}.
Finally, the argument of the function must be computed.
This is by choice of $\headindex_3\geq\headindex_2$ in \descref{MH-ASSIGN}, Lemma~\ref{Lemma:OnceComputedDoesNotChange}, and \descref{SND-D}.

In the end, we must show that the invariants hold in the new state.
The only non-trivial thing is \descref{SND-D}, which holds due to \descref{SND-D}, definition of $\aval$ in \descref{MH-ASSIGN}, definitions of $\evaluate$, and the fact that functions in $\functiondomain$ are deterministic.

\item[Assume]
$\lambda_s=(\headindex_1,\fetchkind,\tid,\controlstate_1\transitionto{\instruction}\controlstate_2)\cdot(\headindex_3,\commitkind,\tid,\anindex)$, $\instruction=\theassume{\anexpr_{\aval}}$.
The proof is similar to the previous case.
The $\commitkind$ transition additionally requires $\evaluate(\tid,\anindex,\anexpr_{\aval})\neq 0$, which holds due to the fact that a similar check in \descref{MH-ASSUME} holds, \descref{SND-D}, definitions of $\evaluate$, the fact that functions in $\functiondomain$ are deterministic.

\item[Load]
$\lambda_s=(\headindex_1,\fetchkind,\tid,\instruction)\cdot(\headindex_2,\loadkind,\tid,\anindex,\anaddr)\cdot(\headindex_3,\commitkind,\tid,\anindex)$, $\instruction=\theload{\areg}{\anexpr_{\anaddr}}$.
Let $\event_1:=(\fetchkind,\tid,\instruction)$, $\event_2:=(\loadkind,\tid,\anindex,\anaddr)$, $\event_3:=(\commitkind,\tid,\anindex)$.

$\powerstate_{\headindex_1}^{s-1}\transitionto{\event_1}\powerstate_{\headindex_1}^s$ holds for the same reasons as before.

Next, we show that $\powerstate_{\headindex_2}^{s-1}\transitionto{\event_2}\powerstate_{\headindex_2}^s$, where this transition is a \descref{POW-EARLY} transition in the early read case of \descref{MH-LOAD} and a \descref{POW-LOAD} transition in the load from memory case.
First, we must show that $\powerstate_{\headindex}^s.\rtstatesmap(\tid).\loaded[\anindex]=\bot$.
This holds by \descref{SND-B3} and the fact that $\mhautomaton$ generates a $\loadkind$ event once and only once for a single fetched load instruction.

Assume the early read case.
This means, $\mearlymemvalue(\tid,\anaddr,\headindex_2)\in\datadomain$.
By \descref{SND-K}, this means, the last fetched store with an unknown address or address of the load is not yet committed, has the address of the load and has the value known.
By \descref{POW-EARLY}, the load can take the value from this store, and \descref{SND-B3} holds in the new state.

Consider the load from memory case.
This means, $\mearlymemvalue(\tid,\anaddr,\headindex_2)=\bot$.
By \descref{SND-K}, this means, there is no earlier fetched store with the same address which is not yet committed.
By \descref{POW-LOAD}, the load can take the value from the last propagated store, and \descref{SND-B3} holds in the new state.

Argumentation for $\powerstate_{\headindex_3}^{s-1}\transitionto{\event_3}\powerstate_{\headindex_3}^s$ is similar to the previous cases.
Additionally, first we must show that $\powerstate_{\headindex}^s.\rtstatesmap(\tid).\loaded[\anindex]\neq\bot$.
This is by $\headindex_3\geq\headindex_2$ (\descref{MH-LOAD}), \descref{SND-B3}.
Second, we must ensure that all preceding instructions accessing the same address $\anaddr$ are committed, and there are no previously fetched instructions with unknown address.
This holds by choice of $\headindex_3$ in \descref{MH-LOAD}, \descref{SND-L}, and \descref{SND-M}.

In the new state, \descref{SND-D} holds by definition of $\aval$ in \descref{POW-LOAD}, definitions of $\evaluate$, \descref{SND-G}, and \descref{SND-K}.
Proofs for the other conditions are simpler.

\item[Store]
$\lambda_s=(\headindex_1,\fetchkind,\tid,\instruction)\cdot(\headindex_3,\commitkind,\tid,\anindex,\coherencekey,\anaddr)\cdot(\headindex_3,\propagatekind,\tid,\tid,\anindex,\anaddr)\cdot(\headindex_4,\propagatekind,\tid_1,\tid,\anindex,\anaddr)\cdots(\headindex_{u+3},\propagatekind,\tid_u,\tid,\anindex,\anaddr)$.
Let $\event_1:=(\fetchkind,\tid,\instruction)$, $\event_3:=(\commitkind,\tid,\anindex,\coherencekey,\anaddr)$, $\event_4:=(\propagatekind,\tid,\tid,\anindex,\anaddr)$, $\event_{j+3}:=(\propagatekind,\tid_j,\tid,\anindex,\anaddr)$ for $j\in\intrange{1}{u}$.

$\powerstate_{\headindex_1}^{s-1}\transitionto{\event_1}\powerstate_{\headindex_1}^s$ holds for the same reasons as before.

$\powerstate_{\headindex_3}^{s-1}\transitionto{\event_3}\powerstate_{\headindex_3}^s$ holds for the same reasons as in the case of a load.
The requirement that the coherence key is unique in \descref{POW-STORE} follows from a similar requirement in \descref{MH-STORE} and \descref{SND-H}.
By \descref{POW-STORE}, the only available transition from $\powerstate_{\headindex_2}^s$ is a propagation of the write to its thread, i.e. $\event_4$, which indeed follows $\event_3$ in $\tau$.
Next, we show that $\event_4$ and further propagate transitions are feasible.

First, \descref{POW-PROP} rule requires the write being propagated to have a coherence key (i.e. to be committed), which holds by choice of $\headindex_j$, $j\in\intrange{3}{u+3}$ in \descref{MH-STORE} and \descref{SND-B2}.
Second, it requires the coherence key of the latest propagated store to be less than the key of the store being propagated.
This is adhered due to the check $\mlastkey(\tid',\anaddr,\headindex)<\coherencekey$ and \descref{SND-H}.

It is easy to see that the inductive statements hold in the new state as well.

\end{description}

Now we prove $\powerstate_{\headindex}^{m}=\powerstate_{\headindex+1}^0$ for all $\headindex\in\intrange{1}{n-1}$.
The equality of $\rtstatesmap$ components immediately follows from \descref{SND-B} inductive statement.

Now we prove $\powerstate_{n}^{m}\in\finalpowerstates$.
\descref{FIN-COMM} holds, because $\powerautomaton$ always emits a commit event for each fetched instruction.

Let us turn to \descref{FIN-LD} property.
First, one should note that $\mhautomaton$ generates $\propagatekind$ events for stores to the same address in each part $\tau_j$ in the ascending order by $\coherencekey$.
This is by \descref{MH-STORE}.
Together with \descref{SND-H}, this means that these events are sorted in $\tau$ in the ascending order by $\coherencekey$.
The rest of the proof of \descref{FIN-LD} is a simple case consideration: whether the loads $\anindex$, $\anindex'$ were done from memory or from a local store early.

\descref{FIN-LD-ST} is proven by a similar case consideration.

\qed
\end{proof}

We call $\alpha$ a \emph{prefix} of $\sigma$ and write $\alpha\sqsubseteq\sigma$ if $\sigma=\alpha\cdot\beta$ for some $\beta$.

\begin{lemma}\label{Lemma:GeneratesAllNormalFormComputations}
$\setcond{\tau\in\computationsOf{\program}{\power}}{\tau\text{ is in normal form of degree $n$}}\subseteq\langOf{\mhautomaton}$.
\end{lemma}
\begin{proof}
Let $\tau=\tau_1\cdots\tau_n\in\computationsOf{\program}{\memorymodel}$ be a normal-form computation, i.e. $\initialpowerstate\transitionto{\tau}\powerstate\in\finalpowerstates$.
We show that there is a sequence of transitions $\initialmhstate\transitionto{\lambda_1}\mhstate_1\transitionto{\lambda_2}\ldots\transitionto{\lambda_m}\mhstate_m\in\finalmhstates$, such that $\tau_\headindex=\secondOf{\projectionOf{(\lambda_1\cdots\lambda_n)}{(\set{\headindex}\times\events)}}$.

Let $\tau_{\headindex}^s:=\secondOf{\projectionOf{(\lambda_1\cdots\lambda_s)}{(\set{\headindex}\times\events)}}$, $\initialpowerstate\transitionto{\tau_1\cdots\tau_{\headindex-1}\cdot\tau_{\headindex}^s}\powerstate_{\headindex}^s\transitionto{}^*\powerstate$.
By induction on $s\in[1,\infty)$ we show the following inductive statements:
\begin{description}
\descitem{CMPL-A} There is a sequence of $s$ transitions: $\initialmhstate\transitionto{\lambda_1}\mhstate_1\transitionto{\lambda_2}\ldots\transitionto{\lambda_s}\mhstate_s$.
\descitem{CMPL-B} For all $\headindex\in\headdomain$: $\tau_{\headindex}=\tau_{\headindex}^s.\overline{\tau_{\headindex}^s}$ for some $\overline{\tau_{\headindex}^s}$.
\descitem{CMPL-C} If $\event_1,\event_2\in\tau$ are two events related to instruction $(\tid,\anindex)$, then $\event_1\in\tau_{\headindex}^s$ for some $\headindex$ iff $\event_2\in\tau_{\headindex'}^s$ for some $\headindex'$.
\descitem{CMPL-D} For each $\tid\in\tiddomain$: $\mcontrolstate(\tid)=\destinationstateOf{\lastOf{\powerstate_{1}^s.\rtstatesmap(\tid).\fetched}}$ (or $\mcontrolstate(\tid)=\initialcontrolstate_{\tid}$ if no instructions were fetched).
\descitem{CMPL-F} For each $\tid\in\tiddomain$, $\areg\in\regdomain$, $\headindex\in\intrange{\mregcomphead(\tid,\areg)}{n}$: $\mregvalue(\tid,\areg)=\evaluate(\tid,\minstrcount(\tid)+1,\areg)$ computed for the state $\powerstate_\headindex^s$.
\descitem{CMPL-F'} For each $\tid\in\tiddomain$, $\areg\in\regdomain$, $\headindex\in\intrange{1}{\mregcomphead(\tid,\areg)-1}$: $\evaluate(\tid,\minstrcount(\tid)+1,\areg)=\bot$.
\descitem{CMPL-G} For each $\tid\in\tiddomain$, $\areg\in\regdomain$, $\headindex\in\intrange{\mregcommhead(\tid)}{n}$: let $\anindex$ be the index of the last instruction in $\powerstate_{\headindex}^s.\rtstatesmap(\tid).\fetched$ writing to $\areg$, then $\anindex\in\powerstate_{\headindex}^s.\rtstatesmap(\tid).\committed$.
\descitem{CMPL-G'} For each $\tid\in\tiddomain$, $\areg\in\regdomain$, $\headindex\in\intrange{1}{\mregcommhead(\tid,\areg)-1}$: let $\anindex$ be the index of the last instruction in $\powerstate_{\headindex}^s.\rtstatesmap(\tid).\fetched$, then $\anindex\not\in\powerstate_{\headindex}^s.\rtstatesmap(\tid).\committed$.
\descitem{CMPL-K} For each $\tid\in\tiddomain$, $\headindex\in\intrange{\massumecommhead(\tid)}{n}$: let $\anindex$ be an index of an $\theassume{}$ instruction in $\powerstate_{\headindex}^s.\rtstatesmap(\tid).\fetched$, then $\anindex\in\powerstate_{\headindex}^s.\rtstatesmap(\tid).\committed$.
\descitem{CMPL-K'} For each $\tid\in\tiddomain$, $\headindex\in\intrange{1}{\massumecommhead(\tid)-1}$: let $\anindex$ be an index of the last $\theassume{}$ instruction in $\powerstate_{\headindex}^s.\rtstatesmap(\tid).\fetched$, then $\anindex\not\in\powerstate_{\headindex}^s.\rtstatesmap(\tid).\committed$.
\descitem{CMPL-L} For each $\tid\in\tiddomain$, $\anaddr\in\addrdomain$, $\headindex\in\headdomain$: let $w:=\powerstate_{\headindex}^s.\propagated(\tid,\anaddr)$.
If $w=\initialstore_{\anaddr}$, $\mmemvalue(\tid,\anaddr,\headindex)=0$.
If $w=(\tid',\anindex')$, $\mmemvalue(\tid,\anaddr,\headindex)=\getvalue(\tid',\anindex')$ computed in $\powerstate_{\headindex}^s$.
\descitem{CMPL-M} For each $\tid\in\tiddomain$, $\anaddr\in\addrdomain$, $\headindex\in\headdomain$: $\guessed{\mlastkey}(\tid,\anaddr,\headindex)<\powerstate_{\headindex}^s.\coherence(\powerstate_{\headindex}^s.\propagated(\tid,\anaddr))=\mlastkey(\tid,\anaddr,\headindex)\leq\guessed{\mlastkey}(\tid,\anaddr,\headindex+1)$.
\descitem{CMPL-N} For each $\tid\in\tiddomain$, $\anaddr\in\addrdomain$, $\headindex\in\headdomain$: let $\anindex\in\naturalnumbers$ be the maximal index, such that $\powerstate_{\headindex}^s.\rtstatesmap(\tid).\fetched[\anindex]$ is a store, $\getaddr(\tid,\anindex)=\anaddr$ in $\powerstate_n^s$.
Let $\anindex'$ be the maximal index, such that $\powerstate_{\headindex}^s.\rtstatesmap(\tid).\fetched[\anindex']$ is a store, $\getaddr(\tid,\anindex')\in\set{\bot,\anaddr}$ in $\powerstate_{\headindex}^s$.
Then $\mearlymemvalue(\tid,\anindex,\headindex)=\bot$ if such $\anindex$ does not exist or $\anindex\in\powerstate_{\headindex}^s.\rtstatesmap(\tid).\committed$.
Otherwise, $\mearlymemvalue(\tid,\anindex,\headindex)=\top$ if $\getaddr(\tid,\anindex')=\bot$ or $\getvalue(\tid,\anindex)=\bot$ in $\powerstate_{\headindex}^s$.
Otherwise, $\mearlymemvalue(\tid,\anindex,\headindex)=\getvalue(\tid,\anindex)$ computed in $\powerstate_{\headindex}^s$.
\descitem{CMPL-P} For each $\tid\in\tiddomain$, $\anaddr\in\addrdomain$, $\headindex\in\intrange{\maddrcommhead(\tid,\anaddr)}{n}$: if $\getaddr(\tid,\anindex)=\anaddr$ in $\powerstate_n^s$ for some $\anindex$, then $\anindex\in\powerstate_{\headindex}^s.\rtstatesmap(\tid).\committed$.
\descitem{CMPL-P'} For each $\tid\in\tiddomain$, $\anaddr\in\addrdomain$, $\headindex\in\intrange{1}{\maddrcommhead(\tid,\anaddr)-1}$: there is $\anindex$ with $\getaddr(\tid,\anindex)=\anaddr$ in $\powerstate_n^s$, such that $\anindex\in\powerstate_{\headindex}^s.\rtstatesmap(\tid).\committed$.
\descitem{CMPL-R} For each $\tid\in\tiddomain$: $\minstrcount(\tid)=\lengthOf{\powerstate_1^s.\rtstatesmap(\tid).\fetched}$.
\end{description}

Base case: $s=1$.
We choose the first (\descref{MH-GUESS}) transition $\initialmhstate\transitionto{\lambda_1}\mhstate_1$, so that the inductive statements hold:
\begin{description}
\item[Guess]
We define $\mmemvalue$ and $\mlastkey$ components of $\mhstate_1$ according to \descref{CMPL-L} and \descref{CMPL-M} requirements.
The other inductive statements trivially hold.
\end{description}

Assume the inductive statements hold for $s$ and $\overline{\tau_{\headindex}^s}\neq\emptysequence$ for some $\headindex\in\headdomain$.
We show they hold for $s':=s+1$.
The proof is done by pointing out an appropriate transition $\mhstate_s\transitionto{\lambda_{s+1}}\mhstate_{s+1}$.
We choose the first possible option out of the following:
\begin{description}

\item[Assignment]
Assume $\event_1\sqsubseteq\overline{\tau_{\headindex_1}^s}$, $\event_3\sqsubseteq\overline{\tau_{\headindex_3}^s}$, where
$\headindex_1<\headindex_3$ ($\headindex_1=\headindex_3$ is possible, but here and further we write strict inequalities for notational convenience),
$\event_1:=(\fetchkind,\tid,\controlstate_1\transitionto{\command}\controlstate_2)$,
$\event_3:=(\commitkind,\tid,\anindex)$,
$\headindex_1=1$,
$\anindex=\minstrcount(\tid)$,
$\command=\theassign{\areg}{\anexpr_{\aval}}$.
Then, as we show next, a \descref{MH-ASSIGN} transition is feasible.

First, $\powerstate_{\headindex_1}^s\transitionto{\event_1}$, therefore, the state of the last fetched instruction in thread $\tid$ in $\powerstate_{\headindex_1}^s$ is $\controlstate_1$.
By \descref{CMPL-D}, $\mcontrolstate(\tid)=\controlstate_1$ too.

Second, we choose $\headindex_2:=\max\setcond{\mregcommhead(\tid,\areg)}{\areg\text{ is read in }\command}$.
It satisfies the requirements from \descref{MH-ASSIGN}.
Note that $\headindex_2\leq\headindex_3$ by \descref{CMPL-F'} and \descref{POW-COMMIT}: an instruction cannot be committed, until its arguments are computed.

Third, we must show that for each register $\areg$ read by the instruction holds $\headindex_3\geq\mregcommhead(\tid,\areg)$ and $\headindex_3$.
This holds by \descref{CMPL-G'}, \descref{CMPL-K'}, and \descref{POW-COMMIT}: an instruction cannot be committed until its data and control dependencies are committed.

In the destination state, \descref{CMPL-F} holds by \descref{CMPL-F} in the source state, definition of $\mregvalue'$ in \descref{MH-ASSIGN} and definitions of $\evaluate$.
The other inductive statements trivially hold.

\item[Assume]
Assume $\event_1\sqsubseteq\overline{\tau_{\headindex_1}^s}$, $\event_3\sqsubseteq\overline{\tau_{\headindex_3}^s}$, where
$\headindex_1<\headindex_3$,
$\event_1=(\fetchkind,\tid,\controlstate_1\transitionto{\command}\controlstate_2)$,
$\event_3=(\commitkind,\tid,\anindex)$, where $\anindex=\minstrcount(\tid)$,
$\headindex_1=1$,
$\anindex=\minstrcount(\tid)$,
$\command=\theassume{\anexpr_{\aval}}$.
Then, a \descref{MH-ASSUME} transition is feasible.

The proof is similar to the proof for the case of assignment.
The \descref{MH-ASSUME} transition additionally requires $\evaluate(\tid,\anexpr_{\aval})\neq 0$.
This holds by \descref{CMPL-F}, definition of $\mregvalue'$ in \descref{MH-ASSIGN} and definitions of $\evaluate$.

The inductive statements trivially hold in the destination state.

\item[Load]
Assume $\event_1\sqsubseteq\overline{\tau_{\headindex_1}^s}$, $\event_2\sqsubseteq\overline{\tau_{\headindex_2}^s}$, $\event_3\sqsubseteq\overline{\tau_{\headindex_3}^s}$, where
$\headindex_1<\headindex_2<\headindex_3$,
$\event_1=(\fetchkind,\tid,\controlstate_1\transitionto{\command}\controlstate_2)$,
$\event_2=(\loadkind,\tid,\anindex,\anaddr)$,
$\event_3=(\commitkind,\tid,\anindex)$,
$\anindex=\minstrcount(\tid)$,
$\command=\theload{\areg}{\anexpr_{\aval}}$.
We show that a \descref{MH-LOAD} transition is feasible.
We point out only differences with respect to the proof for the assignment case.

Assume $\event_2$ was produced by a \descref{POW-EARLY} transition.
This means, the last store writing to $\anaddr$ has its address known and is not committed yet in $\powerstate_{\headindex_2}^s$.
Then, by \descref{CMPL-N}, $\mearlymemvalue(\tid,\anaddr,\headindex_2)\in\datadomain$, and we have $\aval:=\mearlymemvalue(\tid,\anaddr,\headindex_2)$.
Assume $\event_2$ was produced by a \descref{POW-LOAD} transition.
Then, \descref{POW-EARLY} transition was not possible (Lemma~\ref{Lemma:EarlyReadLooksLikeThis}, Lemma~\ref{Lemma:LoadFromMemoryLooksLikeThis}).
This means, there was no in-flight stores to $\anaddr$ in $\powerstate_{\headindex_2}^s$.
Then, by \descref{CMPL-N}, $\mearlymemvalue(\tid,\anaddr,\headindex_2)=\bot$, and we have $\aval:=\mmemvalue(\tid,\anaddr,\headindex_2)$.
In both cases, by \descref{CMPL-N}, \descref{CMPL-L} we have $\mregvalue'$ and $\mregcomphead'$ satisfying \descref{CMPL-F} and \descref{CMPL-F'}.

Additionally, we must show that $\headindex_3\geq\maddrcommhead(\tid,\anaddr)$.
This holds by \descref{CMPL-P'} and \descref{CMPL-N}.

\item[Store]
Assume $u\in\naturalnumbers$, $\event_j\sqsubseteq\overline{\tau_{\headindex_j}^s}$ for $j\in\intrange{1}{u+3}$, where
$\headindex_2=\headindex_3$,
$\event_1=(\fetchkind,\tid,\controlstate_1\transitionto{\command}\controlstate_2)$,
$\event_2=(\commitkind,\tid,\anindex,\coherencekey,\anaddr)$,
$\event_3=(\propagatekind,\tid,\tid,\anindex,\anaddr)$,
$\event_j=(\propagatekind,\tid_j,\tid,\anindex,\anaddr)$ for $j\in\intrange{4}{u+3}$,
$\anindex=\minstrcount(\tid)$,
$\command=\thestore{\anexpr_{\anaddr}}{\anexpr_{\aval}}$.
Assume that there are no other $\propagatekind$ events for $(\tid,\anindex)$ in $\tau$, except for $\event_3\ldots\event_{u+3}$.
We show that a \descref{MH-STORE} transition is feasible.

The requirements to be checked are similar to those in the load case.
The requirement that $\coherencekey$ is not already used holds by \descref{CMPL-M} and the fact that the same requirement in \descref{POW-STORE} is met.

Consider the requirements in \descref{MH-STORE} for generating $\propagatekind$ events.
The requirement that propagation event to thread $\tid$ is generated in the same part as commit is met by assumption $\headindex_3=\headindex_2$.
The requirement $\mlastkey(\tid',\anaddr,\headindex)<\coherencekey\leq\guessed{\mlastkey}(\tid',\anaddr,\headindex+1)$ is met by \descref{CMPL-L}, choice of $\guessed{\mlastkey}$ in the initial transition, and \descref{POW-PROP}.

This means, inductive invariant \descref{CMPL-A} holds for $s+1$.
Also, \descref{CMPL-B} holds by choice of $\event_1\ldots\event_{u+3}$, \descref{CMPL-D} holds trivially.
\descref{CMPL-C} holds by assumption that there are no other $\propagatekind$ events in $\tau$, except for $\event_3\ldots\event_{u+3}$.
\descref{CMPL-F}, \descref{CMPL-F'}, \descref{CMPL-G}, \descref{CMPL-G'} hold as store instruction does not affect register values.
\descref{CMPL-K}, \descref{CMPL-K'} hold as a store instruction is not $\theassume$.
\descref{CMPL-L} holds by definition of $\mmemvalue'$ in \descref{MH-STORE}.
\descref{CMPL-M} holds by definition of $\mlastkey'$ in \descref{MH-STORE}.
\descref{CMPL-N} holds by definition of $\mearlymemvalue'$ in \descref{MH-STORE}.
\descref{CMPL-P}, \descref{CMPL-P'} hold by definition of $\maddrcommhead'$ in \descref{MH-STORE}.
\descref{CMPL-R} hold by definition of $\minstrcount'$ in \descref{MH-STORE}.

\end{description}

Now we must show that one of the cases above always takes place.
Consider the event $\event=\firstOf{\overline{\tau_{1}^s}}$.
By \descref{CMPL-C} and the fact that $\tau\in\computationsOf{\program}{\power}$, it is a $\fetchkind$ event $(\fetchkind,\tid,\anindex,\instruction)$.
Choose the case based on the kind of $\instruction$.
By \descref{NF-A} and \descref{NF-B}, all events related to the instruction $(\tid,\anindex)$ constitute prefixes of $\overline{\tau_{\headindex}^s}$, $\headindex\in\headdomain$.
The requirement $\anindex=\minstrcount(\tid)$ holds by \descref{CMPL-R}.
The requirements like $\headindex_1\leq\headindex_2\leq\headindex_3$ in the load case naturally follow from the fact that $\tau\in\computationsOf{\program}{\power}$.

Assume $\tau_{\headindex}^s=\tau_{\headindex}$ for all $\headindex\in\headdomain$.
Then $\tau_{\headindex}^s\in\finalmhstates$ by choice of $\guessed{\mmemvalue}$ and $\guessed{\mlastkey}$ in $\mhstate_1$ and \descref{CMPL-L}, \descref{CMPL-M}.

\qed
\end{proof}

\begin{lemma}\label{Lemma:MultiheadedAutomatonDoesWhatWeWant}
$\setcond{\tau\in\computationsOf{\program}{\power}}{\tau\text{ is in normal form of degree $n$}}\subseteq\langOf{\mhautomaton(\program)}\subseteq\computationsOf{\program}{\power}$.
\end{lemma}
\begin{proof}
Corollary of Lemmas~\ref{Lemma:GeneratesOnlyCorrectComputations} and~\ref{Lemma:GeneratesAllNormalFormComputations}.
\qed
\end{proof}

\subsection{Checking Cyclicity of the Happens-Before Relation}

We call a happens-before cycle \emph{beautiful}, if it has the following form:
\begin{multline*}
(\tid_1,\anindex_1,\instruction_1)\programorder^*(\tid_1,\anindex_1',\instruction_1')\hoporder\ldots\\\hoporder(\tid_n,\anindex_n,\instruction_n)\programorder^*(\tid_n,\anindex_n',\instruction_n')\hoporder(\tid_1,\anindex_1,\instruction_1).
\end{multline*}
Here, $\hoporder:=(\coherenceorder\cup\sourceorder\cup\conflictorder)$ and $\tid_k\neq\tid_l$ for $k\neq l$.
We call $\theta:=\tid_1\ldots\tid_n$ the \emph{profile} of the cycle.

\begin{example}
The happens-before cycle shown in Figure~\ref{Figure:TraceMP} is beautiful.
\end{example}

\begin{lemma}[\cite{calin2013}]\label{Lemma:BeautifulCycleIsEnough}
A computation $\tau\in\computationsOf{\program}{\power}$ has a happens-before cycle iff it has a beautiful happens-before cycle.
\end{lemma}

Given a cycle profile $\theta$, we define the automaton $\mhautomaton'(\program,\theta)$ as a modification of $\mhautomaton(\program)$ that marks one event in each thread $\tid_j\in\theta$ by $\enterguess$ (identifying $(\tid_j,\anindex_j, *)$) and a later (or the same) event by $\leaveguess$ (identifying $(\tid_j,\anindex_j', *)$, $\anindex_j\leq\anindex_j'$).
Note that $\mhautomaton(\program)$ generates the events in program order, which ensures $(\tid_j,\anindex_j, *)\programorder^*(\tid_j,\anindex_j', *)$.
Technically, $\mhautomaton'(\program,\theta)$ introduces the following changes:
\begin{itemize}
\item The alphabet is $\events':=\events\times\powerset{\set{\enterguess,\leaveguess}}$.
\item The automaton generates only $\loadkind$ and $\propagatekind$ events, as only they are relevant for cycle detection.
\item The $\propagatekind$ events include $\coherencekey$ component of the corresponding $\commitkind$ event.
\end{itemize}

To check $(\tid_j,\anindex_j',*)\hoporder(\tid_{j+1},\anindex_{j+1},*)$, we use an intersection with a regular language $\hopautomaton^{\tid_j,\tid_{j+1}}$.
The language $\hopautomaton^{\tid_1,\tid_2}$ includes a computation $\tau$ iff one or more of the following conditions hold:
\begin{description}
\descitem{H-ST} $(\event_1,\marking_1),(\event_2,\marking_2)\in\tau$, $\leaveguess\in\marking_1$, $\enterguess\in\marking_2$, $\event_1=(\propagatekind,\tid_1,\tid_1,\coherencekey_1,\anaddr)$, $\event_2=(\propagatekind,\tid_2,\tid_2,\coherencekey_2,\anaddr)$, and $\coherencekey_1<\coherencekey_2$.
\descitem{H-SRC} $\tau=\tau_1\cdot(\event_1,\marking_1)\cdot\tau_2\cdot(\event_2,\marking_2)\cdot\tau_3$, $\leaveguess\in\marking_1$, $\enterguess\in\marking_2$, $\event_1=(\propagatekind,\tid_2,\tid_1,\anaddr)$, $\event_2=(\loadkind,\tid_2,\anaddr)$, $\tau_2$ does not contain events $(\propagatekind,\tid_2,*,\anaddr)$.
\descitem{H-CF1} $\tau=\tau_1\cdot(\event_3,\marking_3)\cdot\tau_2\cdot(\event_2,\marking_2)\cdot\tau_3$, $\leaveguess\in\marking_2$, $\event_3=(\propagatekind,\tid_1,\tid_3,\coherencekey_3,\anaddr)$, $\event_2=(\loadkind,\tid_1,\anaddr)$, $\tau_2$ does not contain events $(\propagatekind,\tid_1,*,*,\anaddr)$, $(\event_3,\marking_3)\in\tau_1\cdot\tau_2\cdot\tau_3$, $\marking_3\in\enterguess$, $\event_3=(\propagatekind,\tid_2,\tid_2,\coherencekey_2)$, $\coherencekey_3<\coherencekey_2$.
\descitem{H-CF2} $(\event_1,\marking_1),(\event_2,\marking_2)\in\tau$, $\enterguess\in\marking_1$, $\leaveguess\in\marking_2$, $\event_1=(\loadkind,\tid_1,\anaddr)$, $\event_2=(\propagatekind,\tid_2,\tid_2,\coherencekey_2,\anaddr)$ and there is no $(\event_3,\marking_3)\in\tau$ with $\event_3=(\propagatekind,\tid_3,\tid_3,\coherencekey_3,\anaddr)$ with $\coherencekey_3<\coherencekey_2$.
\end{description}

\begin{lemma}\label{Lemma:ReductionToLanguageEmptiness}
Program $\program$ has a beautiful cycle with profile $\theta=\tid_1\ldots\tid_n$ iff
\begin{equation*}
\mhautomaton'(\program,\theta)\cap\hopautomaton^{\tid_1,\tid_2}\cap\ldots\cap\hopautomaton^{\tid_n,\tid_1}\neq\emptyset.
\end{equation*}
\end{lemma}
Note that $\mhautomaton'(\program,\theta)$ is infinite-state.
To ensure $\mhautomaton'(\program,\theta)$ has finitely many states, we note that the instruction indices are irrelevant for the detection of happens-before cycles ($\minstrcount$ can be dropped), and that the number of different coherence keys that must be stored in the state at any moment is polynomial in the size of $\program$.
Indeed, the $\mlastkey$ and $\guessed{\mlastkey}$ components of the state each store at most $\cardinalityOf{\addrdomain}\cdot\cardinalityOf{\program}\cdot n$ different coherence keys.
Each modification of the $\mlastkey$ component of the state can be extended by a normalization step that would turn coherence keys to consecutive natural numbers starting from zero.
The normalization step must preserve the less-than relation on the keys.
In order for the detection of happens-before cycles to work correctly, the automaton has to remember the coherence keys of marked store events: they must be preserved during normalization.
Altogether, this results into $O(\cardinalityOf{\addrdomain}\cdot\cardinalityOf{\program}^2\cdot n)$ different keys, which is polynomial in the size of $\program$.

\begin{theorem}\label{Theorem:PSpaceComplete}
Robustness against Power is $\pspace$-complete.
\end{theorem}
\begin{proof}
By Theorem~\ref{Theorem:NormalFormComputationsAreEnough}, Lemma~\ref{Lemma:BeautifulCycleIsEnough}, and Lemma~\ref{Lemma:ReductionToLanguageEmptiness}, a program is non-robust iff the equation from Lemma~\ref{Lemma:ReductionToLanguageEmptiness} holds for some $\theta$.
In order to check robustness, we enumerate all profiles $\theta$ and check the equation from Lemma~\ref{Lemma:ReductionToLanguageEmptiness}.
The enumeration can be done in $\pspace$.
By construction and Lemma~\ref{Lemma:MultiheadedIntersection}, the size of the intersection automaton is exponential in the size of the program.
By Lemma~\ref{Lemma:MultiheadedEmptinessComplexity}, language emptiness for it can be checked in $\pspace$ in the size of the program, which gives us the upper bound.

The $\pspace$ lower bound follows from $\pspace$-hardness of SC state reachability.
One can reduce reachability to robustness by inserting an artificial happens-before cycle in the target state.
\qed
\end{proof}

\subsubsection{Acknowledgements.}
The authors thank Parosh Aziz Abdulla, Jade Alglave, Mohamed Faouzi Atig, Ahmed Bouajjani, and Carl Leonardsson for helpful discussions on the Power memory model and the anonymous reviewers for suggestions.
The first author was granted by the Competence Center High Performance Computing and Visualization (CC-HPC) of the Fraunhofer Institute for Industrial Mathematics (ITWM).
The work was partially supported by the DFG project R2M2: Robustness against Relaxed Memory Models.

\bibliographystyle{plain}
\bibliography{cited}

\end{document}